\documentclass[fleqn,12pt]{article}
\usepackage{natbib}
\usepackage{amsthm,amsfonts,amsmath,amssymb,vmargin}
\usepackage[onehalfspacing]{setspace}
\usepackage{enumerate}
\usepackage{pgf,pgfplots,tikz}
\pgfplotsset{compat=1.18}
\usepackage{booktabs}
\usepackage{dsfont}
\usepackage{bm}

\RequirePackage[pagebackref=true]{hyperref}
\hypersetup{colorlinks=true,linkcolor=blue,urlcolor=blue,citecolor=red,
     pdftitle={Interdependent Hitting Times},
    pdfauthor=Jaap H. Abbring and Yifan Yu,
    pdfsubject=Duration Analysis,
	pdfkeywords={duration analysis, identifiability, mixed hitting-time model, stopping game, synchronization game},
    pdfdisplaydoctitle=true,
    bookmarksopen=true}

\setpapersize{A4} \setmargrb{1.25in}{1.25in}{1.25in}{1.25in}

\theoremstyle{plain}
\newtheorem{theorem}{Theorem}
\newtheorem{lemma}{Lemma}

\theoremstyle{definition}
\newtheorem{assumption}{Assumption}

\newtheorem{case}{Case}

\newcommand{\Lap}{{\cal L}}
\newcommand{\s}{z} 
\newcommand{\salt}{\xi} 

\newcommand{\cites}[1]{\citeauthor{#1}'s (\citeyear{#1})}

\newcommand{\val}{v}
\newcommand{\brval}{V}

\newcommand{\err}{\varepsilon}
\newcommand{\bigerr}{\bm{\varepsilon}}

\newcommand{\sX}{{\cal X}}
\newcommand{\X}{X}
\newcommand{\x}{x}
\newcommand{\vX}{\mathbf{X}}
\newcommand{\vx}{\mathbf{x}}
\newcommand{\np}{I} 
\newcommand{\nsims}{Q} 

\newcommand{\epsin}{\delta}

\newcommand{\R}{\mathbb{R}}
\newcommand{\Rp}{[0,\infty)}
\newcommand{\Rpp}{(0,\infty)}

\newcommand{\barYw}{\widehat{Y}}

\newcommand{\MCone}{A}
\newcommand{\MCtwo}{B}
\newcommand{\MCthree}{C}
\newcommand{\MCfour}{D}

\newcommand{\computeplatform}{The final row of each panel gives the average computation time per Monte Carlo iteration  in seconds. This is a conservative estimate of typical computing times, as we only parallelized Monte Carlo iterations and not the multistart estimation routine within each iteration. Moreover, simulations were run using MATLAB 2025b on a Windows Server 2019 virtual machine with 60 allocated virtual CPUs hosted on AMD EPYC 9754 processors, which offers effective parallelization, but only modest per-core performance.}

\begin{document}

\title{Interdependent Hitting Times}
\author{Jaap H. Abbring\thanks{%
Department of Econometrics \& OR, Tilburg University, P.O. Box 90153, 5000 LE Tilburg, The Netherlands; and CEPR. E-mail \href{mailto:jaap@abbring.org}{jaap@abbring.org}. Web: \href{http://jaap.abbring.org}{jaap.abbring.org}.}\and 
Yifan Yu\thanks{%
School of Economics, Zhejiang University,  Hangzhou, China. E-mail: \href{mailto:yifanyu@zju.edu.cn}{yifanyu@zju.edu.cn}. 
\newline\newline
We thank \'{A}ureo de Paula and participants in numerous conferences and seminars for helpful discussion. This research was financially supported by the Dutch Research Council (NWO) through Vici grant 453-11-002. It used the Dutch national e-infrastructure with the support of the SURF Cooperative and two Small Compute grants (EINF-10757 and EINF-15783) in the Computing Time on National Computer Facilities programme of NWO.
 }}


\maketitle
\vspace*{-1em}

\begin{abstract}
This paper studies interdependent durations as equilibrium outcomes of a synchronization game, a continuous-time stopping game in which the incentive to stop increases when other players stop.  We allow the payoffs to vary with both common shocks and observed and unobserved agent characteristics. The common shocks follow a spectrally negative L\'{e}vy process, a semiparametric process that includes Brownian motion as a special case but may also have jumps. We show that equilibrium outcomes can be represented as interdependent hitting times and use this to establish the game's nonparametric identification from data on stopping times and covariates. We develop maximum simulated likelihood and method of simulated moments estimators and evaluate their finite-sample and computational performance in Monte Carlo experiments. The results provide a tractable framework for identifying and estimating synchronization games from interdependent duration data.
\end{abstract}

\thispagestyle{empty}

\newpage\setcounter{page}{1}

\section{Introduction}

Agents are often observed to synchronize their stopping choices. This could be because they complement each other in an activity, so that one agent abandoning it triggers the exit of others, but also because they face common shocks or have related characteristics. For example, mall operators use anchor stores, such as flagship stores, to generate consumer traffic \citep{vitorino2012empirical}. One anchor store's exit may reduce profits of the mall's remaining anchor stores and make them leave as well. Alternatively, stores may exit together because their profits are hit by common negative shocks, such as a mass layoff that reduces local incomes. Similarity of stores within malls, for example because they sort on quality into malls, amplifies synchronization driven by complementarities and common shocks. If stores differ across malls, such sorting may itself create the appearance of synchronization. 

Distinguishing these sources of synchronization matters because they represent different mechanisms and have distinct policy implications. In particular, the causal effects due to complementarities are instances of ``endogenous effects'' \citep{manski_identification_1993}. These can create ``social multipliers'' that amplify other determinants of behavior and thereby change the effects of policy interventions. For example, if anchor stores complement each other, avoiding the first exit may be particularly effective, and mall operators may want to direct their efforts there.

In this paper, we develop an econometric framework to study these problems. We specify a continuous-time optimal stopping game with strategic complementarities; that is, in which the incentive to stop increases when other players stop. We allow for both a semiparametric (spectrally negative L\'{e}vy) process of common shocks and nonparametric observed and unobserved heterogeneity. Our key contribution is to show that the game's equilibrium outcomes can be represented in terms of the times at which the L\'{e}vy process hits heterogeneous, interdependent thresholds. This allows us to derive nonparametric identification results for the L\'{e}vy process and the  thresholds that characterize equilibrium and to construct computationally feasible estimators. This yields a full econometric development of a standard class of synchronization games. It also shows how results for mixed hitting-time models can be extended to multivariate settings and applied to stopping games with heterogeneous agents facing ongoing uncertainty.

Our contribution is closest to \citet{de2009inference}, which proposed a similar model and applied it to study the synchronized desertion of Union Army soldiers during the American Civil War. This model includes strategic complementarities, reflecting the negative effects of desertion on the morale and combat capabilities of a company’s remaining soldiers, and common payoff shocks inflicted by, for example, sustained artillery fire. A key assumption is that the shocks enter as correlated {\em privately observed} Brownian motions. In equilibrium, players form beliefs about other players' stochastically evolving payoff states based on their own Brownian motion observations and the commonly observed survival of the other players. These beliefs will generally be that there is a nonzero hazard of any remaining players exiting, which, if there are complementarities, implies a nontrivial expected value loss term in the Hamilton-Jacobi-Bellman equation for the best-response problem. This is recognized by \citeauthor{de2009inference}'s Proposition 1, which covers the best-response problem for given exit hazards of the other players, but is not carried through in the equilibrium analysis in \citeauthor{de2009inference}'s Proposition 2. It may be possible to characterize equilibrium strategies in \citeauthor{de2009inference}'s game as fixed points of Proposition 1's best-response mapping with consistency between players' beliefs and the exit hazards induced by the resulting strategies, but this is nontrivial. Our complete-information approach removes this complication and facilitates a rigorous and full equilibrium characterization in terms of threshold rules. 

In turn, this simplifies the econometric analysis, so that we can provide point identification results for a model with a semiparametric shock process and nonparametric observed and unobserved heterogeneity. In terms of \citeauthor{de2009inference}'s application to desertion, this not only allows us to more flexibly distinguish complementarities and common shocks, but also to measure whether their effects are particularly strong because soldiers sort into companies based on unobservables like courage. Moreover, we not only provide and evaluate a method of simulated moments (MSM) estimator inspired by \citeauthor{de2009inference}'s minimum simulated distance approach, but also a maximum simulated likelihood (MSL) estimator based on a custom recursive simulator that exploits our model's threshold structure.

Later contributions pursued tractability by specifying deterministic payoff processes. \citet{qe18:honoredepaula} studied complementarities in retirement with a model in which couples cooperatively pick the Nash bargaining solutions to their retirement problems. The assumption of deterministic payoffs, with appropriate functional form choices, allowed \citeauthor{qe18:honoredepaula} to apply identification results for generalized accelerated failure-time models \citep{ridder_non-parametric_1990}. \citet{honore_interdependent_2010} studied a noncooperative variant of this game with complete information and \citet{lin_liu_2021} developed one with incomplete information. In contrast, we study games with ongoing payoff uncertainty. 

We also contribute to econometric duration analysis. Traditional duration models are typically specified using hazard rates. They include multivariate models that allow durations to be dependent through both direct effects of one duration on the hazards of others and related unobserved characteristics \citep{ecma03:abbringvandenberg}. Unlike these hazard models, our interdependent hitting-times model links observed durations to economic primitives. This allows us to separate the mechanisms underlying synchronization and facilitates counterfactual analysis. 

The applied literature on dynamic games often assumes discrete time.  We work in continuous time because this is common in traditional duration analysis. Moreover, continuous time simplifies the theoretical analysis \citep{honore_interdependent_2010} and may address some of the computational problems that plague discrete-time games \citep{doraszelski2012avoiding}.  Indeed, \citet{res16:arcidiaconoetal} specified a dynamic discrete-choice game in which agents can only take actions at (almost surely distinct) Poisson times and discussed how standard estimation methods can be applied to this game. In contrast, the process underlying our game is not a Poisson process but a more general L\'{e}vy process. Moreover, we provide identification results that are not available from that work. This is all facilitated by a close link to mixed hitting-time (MHT) models. \citet{are10:abbring,abbring_mixed_2012} studied such models and discussed their application to single-agent optimal stopping problems. Our paper extends this to games.

Our theoretical analysis builds on the literature on continuous-time optimal stopping models. In economics, these models are used to study the optimal timing of irreversible decisions under uncertainty. They are closely related to option-pricing models in finance and often referred to as real-options models. \cite{dixit_investment_1994} and \cite{stokey_economics_2009} analyzed and reviewed various models based on Brownian motion; \cite{cont_financial_2004}, \cite{kyprianou_introductory_2006}, and \cite{boyarchenko_irreversible_2007} provided extensions to general L\'{e}vy processes. 

Our model is a real-options \emph{game}. Real-options games are multivariate extensions of single-agent real-options models with strategic interactions. They introduce ongoing uncertainty into stopping games with deterministic payoffs, such as \cites{fudenberg_preemption_1985} pre-emption game and \cites{fudenberg_theory_1986} war of attrition. Our theoretical analysis is closest to \cite{murto_exit_2004}, which studied a war of attrition in which two firms decide when to abandon a declining, uncertain market. \citeauthor{murto_exit_2004}'s model is a two-player stopping game with strategic substitutes, with uncertainty modelled using Brownian motion; ours is a multi-player stopping game with strategic complements driven by a more general L\'{e}vy process. Our contribution is to characterize equilibrium outcomes of such games in terms of interdependent hitting times, thereby turning a class of real-options games into econometric models for duration data.

The paper is organized as follows. In the next section, we highlight  the paper's main ideas using a simple example of an empirical synchronization game. Section \ref{s:theory} introduces the general synchronization game and characterizes its equilibrium outcomes as interdependent hitting times. Section \ref{s:identification} discusses the model's econometric implementation and identification. Section \ref{s:estimation} proposes and evaluates estimators. Section \ref{s:conclusion} concludes by discussing the application to other games. Appendices \ref{app:characterization} and \ref{app:identGen} offer proofs omitted from the main text. An online supplement collects supporting material \citep{ecsgsupp26:abbringyu}.

\section{Example}
\label{s:example}

We first illustrate this paper's main ideas with a simple example of an empirical synchronization game, applied to the analysis of spillovers between anchor stores in shopping malls. Mall operators use anchor stores, for example big department or flagship stores, to generate consumer traffic. Other stores (including other anchor stores) may benefit from the presence of  an anchor store, even if it sells close substitutes. Indeed, \cite{vitorino2012empirical} estimated a static model of anchor store entry, using a sample of regional shopping centers in the United States, and found that positive spillovers dominate competition between anchor stores. Unlike \citeauthor{vitorino2012empirical}'s model, our example game of anchor store survival is  {\em dynamic}, which, for example, would allow us to explore the effects of uncertainty and option values. 

\subsection{A Simple Synchronization Game}

Consider a shopping mall with, at time $t=0$, two anchor stores, $A$ and $B$. The mall is dwindling: No new anchor stores will find it profitable to enter and the incumbent stores are each choosing the best time $t\in\Rp$ to permanently close. Specifically, if both are still active at time $t$, they decide simultaneously, in a ``Joint'' decision node, whether to exit or not. If store $i$ exits but the other store, which we denote with $-i$, stays in the Joint node at time $t$, store $-i$ can respond immediately by exiting in a subsequent ``Lone'' decision node; $i=A,B$ (in Section \ref{s:theory}, we make this more precise). If store $-i$ nonetheless continues alone, it can exit in a Lone decision node at a future time in $(t,\infty)$, or stay forever.

Following the exit decisions at time $t$, payoffs accrue. If store $i$ has closed in $[0,t]$, it receives the outside payoff of $0$. If store $i$ is still active, it receives a flow profit $R^J_i-C^J_i\exp(\gamma_i Y_t)$ if store $-i$ is still active and $R^L_i-C^L_i\exp(\gamma_i Y_t)$ if store $-i$ has closed. Here, $Y_t=\mu t + \sigma W_t$ is a common driver of profits that is not affected by the stores' choices (the mall's ``external'' state), with $\{W_t\}$ a standard Brownian motion, $\mu\geq 0$ a drift parameter, and $\sigma>0$ a dispersion parameter. We let  $R_i^J>R_i^L>0$ and $C_i^L>C_i^J>0$, so that store $i$ finds it less profitable to be active once store $-i$ has closed, and $\gamma_i>0$, so that profits decrease in $Y_t$.

The stores have complete and almost perfect information: When deciding on exit at time $t$, store $i$ knows the game's parameters, the history $\{Y(\tau);0\leq\tau\leq t\}$ of the mall's profitability state, and the stores' actions in all preceding nodes. They use pure stopping strategies: rules that tell them for each possible history to either exit or continue. Because $\{Y\}$ is a Markov process, it is natural to restrict attention to Markov strategies, which are only contingent on payoff-relevant state variables. Apart from the game's parameters, at time $t$, these are simply the current $Y_t$ and whether both stores are still active. Thus, a Markov strategy for store $i$ can be represented by stopping sets ${\cal Y}^J_i\subseteq\R$ and ${\cal Y}^L_i\subseteq\R$: If store $i$ employs a strategy $({\cal Y}^J_i, {\cal Y}^L_i)$, it exits in a Joint decision node at time $t$ if and only if $Y_t\in{\cal Y}^J_i$ and in a Lone decision node at time $t$ if and only if $Y_t\in{\cal Y}^L_i$.

We focus on Markov-perfect equilibria, subgame-perfect equilibria in Markov strategies. In particular, each store $i$ picks a strategy $({\cal Y}^J_i, {\cal Y}^L_i)$ that maximizes her expected flow of payoffs discounted at a rate $\rho_i>\psi(\gamma_i)\equiv\mu\gamma_i+\sigma^2\gamma_i^2/2>0$ given store $-i$'s strategy $({\cal Y}^J_{-i}, {\cal Y}^L_{-i})$. We can analyze these best responses and their equilibria using well-known approaches to the corresponding single-agent optimal stopping problem \citep[\emph{e.g.},][Section 6]{stokey_economics_2009}. 

\subsection{Equilibrium}

Because exit is permanent and there is no entry, we can solve for our game's equilibria backwards, starting with Lone subgames.

\subsubsection{Lone subgames}

Consider a subgame starting in time $t$'s Lone node with store $i$ active and store $-i$ closed. Because store $-i$ has closed for good, store $i$ faces a single agent stopping problem in this subgame, with a flow of payoffs $R^L_i-C^L_i\exp(\gamma_i Y_t)$ until it exits (if ever), which returns a zero payoff. This problem has a textbook solution \citep[][Section S2.3]{ecsgsupp26:abbringyu}: Store $i$ closes as soon as $Y_t$ increases above the threshold
\begin{equation*}
\overline{Y}^L_i\equiv\gamma_i^{-1}\ln\left(\frac{R_i^L}{C_i^L}\cdot\frac{\Lambda(\rho_i)}{\Lambda(\rho_i)-\gamma_i}\cdot 
\frac{\rho_i-\psi(\gamma_i)}{\rho_i}\right),
\end{equation*} 

\noindent where $\Lambda(\rho_i)\equiv\left(-\mu+\sqrt{\mu^2+2\rho_i\sigma^2}\right)/\sigma^2>\gamma_i$ solves $\psi\left(\Lambda(\rho_i)\right)=\rho_i$. This threshold rule's properties are intuitive. Store $i$ exits faster if it generates less profit in any given state; that is, if $R^L_i/C^L_i$ is smaller.  If it is myopic ($\rho\rightarrow\infty$), it is not interested in any upward potential in staying active and closes as soon as it no longer makes a short term profit: $R^L_i-C^L_i\exp(\gamma_i\overline{Y}^L_i)=0$. It also applies this rule in the limit $\sigma\downarrow 0$, where there {\em is} no upward potential to care about. Away from these extremes, stores will tolerate short term losses that are small enough relative to the mall's upward potential:  $R_i^L-C_i^L\exp\left(\gamma_i\overline{Y}^L_i\right)<0$. They are more tolerant of such losses ($\overline{Y}^L_i$ is higher) if they are more patient ($\rho$ is lower) or see more upward potential ($\sigma$ is higher).  

In any Markov-perfect equilibrium, ${\cal Y}^L_A=[\overline{Y}^L_A,\infty)$ and ${\cal Y}^L_B=[\overline{Y}^L_B,\infty)$.\footnote{It would also be optimal for stores to apply the threshold rule strictly and use open stopping sets ${\cal Y}^L_i=(\overline{Y}^L_i,\infty)$. This would give the same duration distributions.} This settles the equilibrium stopping behavior in the Lone decision nodes, in which only one of the stores is still active. 

\subsubsection{Joint subgames and equilibrium outcomes}

Stopping behavior in the Joint decision nodes is more complicated, because of the positive spillovers between the stores. Equilibrium strategies need not specify threshold rules for Joint nodes. Nevertheless, because Brownian motion is continuous, provided that the lower end of the stopping set $\inf {\cal Y}^J_i\geq 0$, the first time $\{Y\}$ hits ${\cal Y}^J_i$ equals the first time $T\left(\inf {\cal Y}^J_i\right)$ it hits $\inf {\cal Y}^J_i$, where $T\left(y\right)\equiv\inf\{t\geq 0\mid Y_t\geq y\}$.\footnote{We use the convention that $\inf \emptyset \equiv \infty $. In particular, we set $T(y)=\infty$ if $\{Y\}$ never hits $y$.} So, for the purpose of econometrically modelling exit times, it is not very restrictive to focus on equilibria in threshold strategies. 

To characterize such equilibria, first consider the auxiliary single agent stopping problem of a store $i$ that receives a flow of payoffs $R^J_i-C^J_i\exp(\gamma_iY_t)$ until it exits, which returns zero payoff. This problem is identical to the problem of a lone surviving store studied above, with a change of parameters from $(R^L_i,C^L_i)$ to $(R^J_i,C^J_i)$. The stopping problem of a lone surviving store is at one extreme that is particularly favorable to closing; this new auxiliary problem is at the other extreme, which is most favorable to remaining active. The solution to the auxiliary problem is again a threshold rule, with threshold
\begin{equation*}
\overline{Y}^J_i\equiv\gamma_i^{-1}\ln\left(\frac{R_i^J}{C_i^J}\cdot\frac{\Lambda(\rho_i)}{\Lambda(\rho_i)-\gamma_i}\cdot 
\frac{\rho_i-\psi(\gamma_i)}{\rho_i}\right)>\overline{Y}^L_i.
\end{equation*} 
With these thresholds for the auxiliary stopping problem in hand, we can analyze our game's Markov-perfect equilibria.  Denote $\min \overline{Y}^J\equiv\min\{\overline{Y}_A^J,\overline{Y}_B^J\}$ and $\max \overline{Y}^L\equiv\max\{\overline{Y}_A^L,\overline{Y}_B^L\}$. We distinguish two cases.

\def\ylow{0.7}
\def\ymid{1.0}
\def\yhigh{1.8}
\def\yfour{1.7}

\def\simBM{  1}
\def\simmu{  1}
\def\simsigma{  1}
\def\simCP{  0}
\def\simlambda{1.00}
\def\simmshock{1.00}
\def\simtmin{  0}
\def\simtmax{  1}
\def\simhmin{-0.15}
\def\simhmax{2.24}
\def\simlegx{0.780000}
\def\simlegya{0.821335}
\def\simlegyb{1.992376}
\def\simhtaa{0.537500}
\def\simhtxaa{0.54}
\def\simhtab{0.775000}
\def\simhtxab{0.78}
\def\simhtac{      }
\def\simhtxac{   }
\def\simhtad{      }
\def\simhtxad{   }
\def\simhtba{0.232000}
\def\simhtxba{0.23}
\def\simhtbb{0.271000}
\def\simhtxbb{0.27}
\def\simhtbc{0.566000}
\def\simhtxbc{0.57}
\def\simhtbd{0.562000}
\def\simhtxbd{0.56}

\def\ylowseq{0.5}
\def\yhighseq{2}
\begin{figure}[tb]
	\begin{tikzpicture}
	\begin{axis}[width=4.8in,height=3in,
	name=sim1,
	ylabel={},
	xlabel={$t\rightarrow$},
	xlabel style={at={(ticklabel cs:1)},yshift=0.8em,anchor=south east, gray},
	axis x line=bottom,
	every outer x axis line/.append style={-,color=gray,line width=1pt},
	axis y line=left,
	every outer y axis line/.append style={-,color=gray,line width=1pt},
	xtick={\simtmin,\simhtbb,\simhtbd},
	xticklabels={$\simtmin$,$T^{1}$,$T^{2}$,},
	ytick={\simhmin,0,\ylowseq,\ymid,\yfour,\yhighseq},          
	yticklabels={,0,$\overline{Y}^L_i$,$\min\overline{Y}^J=\overline{Y}^J_i$,$\max\overline{Y}^L=\overline{Y}^L_{-i}$,$\overline{Y}^J_{-i}$},
	major tick length=2.5mm,
	every tick/.append style={line width=0.75pt},
	ymin=\simhmin,
	ymax=\simhmax,
	scaled ticks=false,
	/pgf/number format/precision=2,
	/pgf/number format/set thousands separator={}]
	\addplot[color=black,line width=0.05pt] table[x=t,y=p2,col sep=comma]{compute/mhtmsim.csv};
	\draw[color=green!75,line width=1pt,solid] (axis cs:\simtmin,\ymid) -- (axis cs:\simtmax,\ymid);
	\draw[color=red!75,line width=1pt,solid] (axis cs:\simtmin,\yfour) -- (axis cs:\simtmax,\yfour);
	\draw[color=green!75,line width=1pt,dashed] (axis cs:\simtmin,\yhighseq) -- (axis cs:\simtmax,\yhighseq);
	\draw[color=red!75,line width=1pt,dashed] (axis cs:\simtmin,\ylowseq) -- (axis cs:\simtmax,\ylowseq);
	\draw[color=green!75,line width=1pt,dotted] (axis cs:\simhtbb,\ymid) -- (axis cs:\simhtbb,\simhmin);
	\draw[color=red!75,line width=1pt,dotted] (axis cs:\simhtbd,\yfour) -- (axis cs:\simhtbd,\simhmin);
	\end{axis}
	\end{tikzpicture}
\caption{Sequential Exit\label{fig:seqexit}}

\vspace*{5pt}
{\footnotesize
Note: This figure plots one realization of $\{Y\}$ and thresholds that reflect that heterogeneity dominates complementarity (Case \ref{case:sequential}): $\overline{Y}_i^L< \overline{Y}_i^J <\overline{Y}_{-i}^L<\overline{Y}_{-i}^J$ for $i\in\{A,B\}$. In equilibrium, store $i$ exits at $T^1=T(\min\overline{Y}^J)$ and store $-i$ at time $T^2= T(\max\overline{Y}^L)>T^J$.
}
\end{figure}
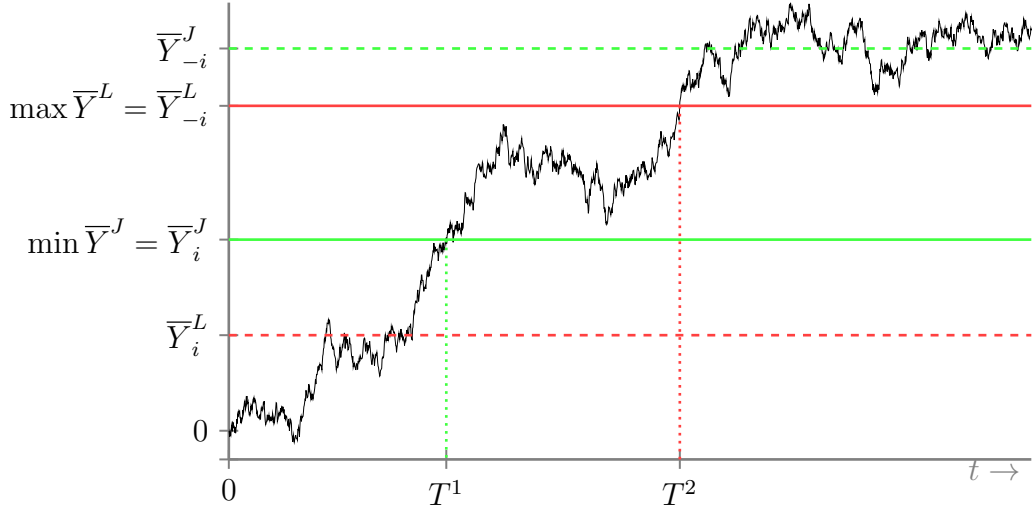

\begin{case}  {\bf Heterogeneity dominates complementarity}\label{case:sequential}\\
If $\min \overline{Y}^J<\max \overline{Y}^L$, $\overline{Y}_i^L< \overline{Y}_i^J <\overline{Y}_{-i}^L<\overline{Y}_{-i}^J$ for $i\in\{A,B\}$; see Figure \ref{fig:seqexit}.  Store  $-i$ will always continue at $Y_t<\overline{Y}_{-i}^L$ in equilibrium. Therefore, store $-i$ will be around for at least as long as $Y_t\leq\overline{Y}_i^J<\overline{Y}_{-i}^L$, so that store $i$'s best response solves the auxiliary decision problem with Joint payoffs and ${\cal Y}_i^J=\left[\overline{Y}_i^J,\infty\right)$ in any equilibrium. In contrast, store $-i$'s equilibrium exit behavior in Joint nodes when $Y_t> \overline{Y}_{-i}^L$ is not uniquely determined: At such $Y_t$, store $i$ leaves and store $-i$ wishes to leave as well, but it can either achieve this by exiting in the Joint node or by continuing in the Joint node and exiting immediately after in the Lone node. However, these multiple equilibria all give the same outcome: The first store exits when $Y_t$ hits $\min\overline{Y}^J$, at time $T^1=T(\min\overline{Y}^J)$, and the second (lone surviving) store will leave when $Y_t$ subsequently hits $\max\overline{Y}^L$, at $T^2=T(\max\overline{Y}^L)$. We refer to this as ``sequential exit.'' 
\end{case}

\def\ylowsim{0.3}
\def\yhighsim{2.1}
\begin{figure}[tb]
	\begin{tikzpicture}
	\begin{axis}[width=4.8in,height=3in,
	name=sim2,
	anchor={origin},
	ylabel={},
	xlabel={$t\rightarrow$},
	xlabel style={at={(ticklabel cs:1)},yshift=0.8em,anchor=south east, gray},
	axis x line=bottom,
	every outer x axis line/.append style={-,color=gray,line width=1pt},
	axis y line=left,
	every outer y axis line/.append style={-,color=gray,line width=1pt},
	xtick={\simtmin,\simhtba,\simhtbc},
	xticklabels={$\simtmin$,,$T^{1}$,},
	ytick={\simhmin,0,\ylowsim,\ylow,\yhigh,\yhighsim},          
	yticklabels={,0,$\overline{Y}^L_i$,$\max\overline{Y}^L=\overline{Y}^L_{-i}$,$\min\overline{Y}^J=\overline{Y}^J_i$,$\overline{Y}^J_{-i}$},
	major tick length=2.5mm,
	every tick/.append style={line width=0.75pt},
	ymin=\simhmin,
	ymax=\simhmax,
	scaled ticks=false,
	/pgf/number format/precision=2,
	/pgf/number format/set thousands separator={}]
	\addplot[color=black,line width=0.05pt] table[x=t,y=p2,col sep=comma]{compute/mhtmsim.csv};
	\draw[color=green!75,line width=1pt,solid] (axis cs:\simtmin,\yhigh) -- (axis cs:\simtmax,\yhigh);
	\draw[color=red!75,line width=1pt,solid] (axis cs:\simtmin,\ylow) -- (axis cs:\simtmax,\ylow);
	\draw[color=green!75,line width=1pt,dashed] (axis cs:\simtmin,\yhighsim) -- (axis cs:\simtmax,\yhighsim);
	\draw[color=red!75,line width=1pt,dashed] (axis cs:\simtmin,\ylowsim) -- (axis cs:\simtmax,\ylowsim);
	\draw[color=green!75,line width=1pt,dotted] (axis cs:\simhtbc,\yhigh) -- (axis cs:\simhtbc,\simhmin);
	\draw[color=red!75,line width=1pt,dotted] (axis cs:\simhtba,\ylow) -- (axis cs:\simhtba,\simhmin);
	\draw[color=blue!75,line width=5pt,solid] (axis cs:\simhtba,\simhmin) -- (axis cs:\simhtbc,\simhmin) node[pos=0.5, color=blue, anchor = south]{$T(\overline{Y})\in$};
	\draw[color=blue!75,line width=5pt,solid] (axis cs:\simtmin,\ylow) -- (axis cs:\simtmin,\yhigh)  node[pos=0.5, color=blue, anchor = west]{$\ni\overline{Y}$};
	\end{axis}
	\end{tikzpicture}
\caption{Simultaneous Exit\label{fig:simexit}}

\vspace*{5pt}
{\footnotesize
Note: This figure plots one realization of $\{Y\}$ and thresholds that reflect that complementarity dominates heterogeneity (Case \ref{case:synchronization}): $\overline{Y}_i^L< \overline{Y}_{-i}^L <\overline{Y}_i^J<\overline{Y}_{-i}^J$ for $i\in\{A,B\}$. In equilibrium, both stores exit at the time $T(\overline{Y})$ at which $\{Y\}$ hits $\overline{Y}\in\left[\max\overline{Y}^L,\min\overline{Y}^J\right]$. Both firms would be best off in  the ``efficient'' equilibrium in which they exit at $T^1=T(\min\overline{Y}^J)$.
}
\end{figure}
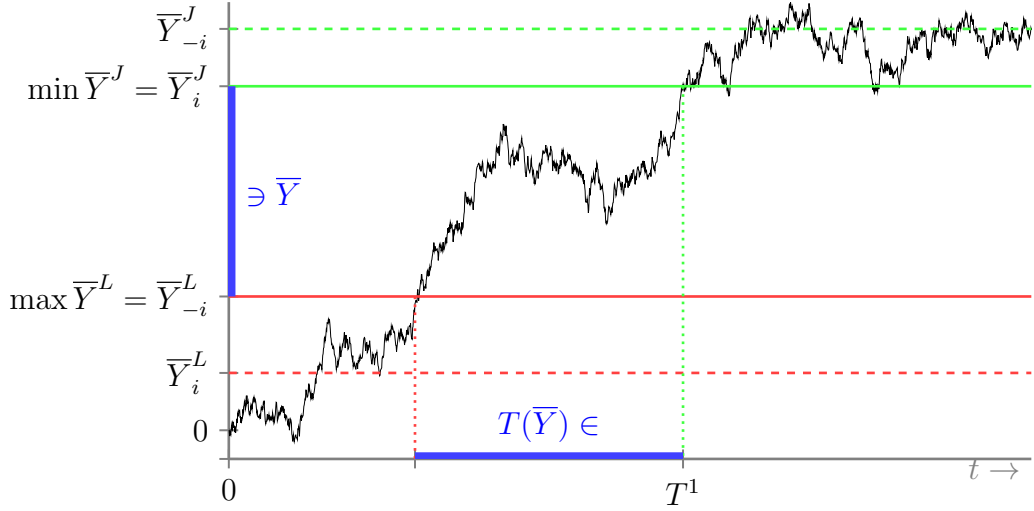

\begin{case} {\bf Complementarity dominates heterogeneity}\label{case:synchronization}\\
If $\min \overline{Y}^J\geq\max \overline{Y}^L$, the complementarity of the payoffs is large enough, relative to the differences between the two stores, to ensure that neither store would continue alone once the other store has closed. Consequently, both stores exit at the same time. We refer to this as ``simultaneous exit.''  Figure \ref{fig:simexit} shows one such scenario, in which $\overline{Y}_i^L< \overline{Y}_{-i}^L <\overline{Y}_i^J<\overline{Y}_{-i}^J$. There are multiple equilibria and, unlike in the sequential exit case, these correspond to multiple outcomes. In general, simultaneous exit at any threshold $\overline{Y}\in\left[\max\overline{Y}^L,\min\overline{Y}^J\right]$ can be supported as an equilibrium outcome. However, a natural equilibrium refinement would pick those equilibria in which both stores enjoy their positive spillovers as much as they can and only exit when $Y_t$ hits the highest among these thresholds, $\min\overline{Y}^J$, at time $T^1=T(\min\overline{Y}^J)$. In our econometric analysis, we apply this refinement. 
\end{case}

\subsection{Empirical Implementation}

Suppose we have a sample from a population of regional shopping malls, with data on exit times $(T^{J},T^{L})$ and possibly the identities of the first and second store exiting and some background characteristics of each store. We assume that the  observed exit times are equilibrium outcomes of the previous section's game, with primitives that may vary across shopping malls. Specifically, we take the distribution of $\{Y\}$ to be common across shopping malls, but allow the payoff parameters $(R^J_A,C^J_A,R^L_A,C^L_A,\gamma_A,\rho_A)$ and $(R^J_B,C^J_B,R^L_B,C^L_B,\gamma_B,\rho_B)$, and therefore the thresholds that characterize the equilibrium strategies, to vary across malls with the stores' observed and unobserved (by the econometrician) characteristics. 

Complementarities in a store's payoffs from positive spillover effects, the common shocks to these payoffs through $\{Y\}$, and observed and unobserved heterogeneity all affect the observed exit times and their dependence between anchor stores within a shopping mall. Specifically, stores may exit together--- ``synchronize'' their exit behavior--- because of strong complementarities in the payoffs, common shocks to those payoffs, or sorting of similar stores into shopping malls (so that co-located stores have similar profitabilities, but potentially very different payoffs from those in other shopping malls). It is not obvious to what extent the data allow us to uniquely determine (``identify'') the model's primitives and, in particular, to separate these sources of synchronization, even in the absence of sampling error.

The key insight we leverage in this paper is that equilibrium exit times are first hitting times of heterogeneous thresholds. Specifically, in both equilibrium cases, the time at which the first store in a mall closes is the first time that $Y_t$ hits $\min\overline{Y}^J$, so that \cites{abbring_mixed_2012} identification results for the mixed hitting-time model, and \cites{jem21:abbringsalimans} stronger results for the case that $\{Y\}$ has a nontrivial Gaussian component,  can be applied to the identification of the parameters $\mu$ and $\sigma^2$ of $\{Y\}$ and the determinants of the threshold $\min\overline{Y}^J$. 

If the data are stratified into groups of shopping malls that share the same payoff parameters (where the parameters may continue to be different across stores within each mall), then \citeauthor{abbring_mixed_2012}'s Theorem 2 ensures identification of the distributions of the payoff shocks $\{Y\}$ and the threshold $\min\overline{Y}^J$. Under additional assumptions--- \emph{e.g.},  that stores have been randomly assigned to shopping malls, so that their store characteristics and thresholds are identically and independently distributed--- this identifies the distribution of $(\overline{Y}_A^J,\overline{Y}_B^J)$. If we want to learn about the decisions of lone survivors, and thus $\overline{Y}_A^L$ and $\overline{Y}_B^L$,  as well, we need to bring in data on lone surviving stores, which is complicated by the fact that these are only observed for the selected group of stores that exit sequentially:  $\min \overline{Y}^J<\max \overline{Y}^L$.

Without stratified data, identification can use variation with the stores' observed characteristics, using variants of \citeauthor{jem21:abbringsalimans}'s Theorem 1. As in \citeauthor{abbring_mixed_2012}'s examples of single-agent optimal stopping problems, this requires that we structure the dependence of the primitives on the observed and unobserved characteristics to give  sufficiently convenient expressions for the implied thresholds. We fully develop all this in Section \ref{s:identification}, for Section \ref{s:theory}'s general model.

\section{Theory}
\label{s:theory}

The anchor store example introduces the class of games that we study and highlights the key role of interdependent hitting times in their empirical analysis. Sections \ref{ss:game}--\ref{ss:equilibrium} extend the simple example to general profit (or, rather, utility) functions and an external state that follows a spectrally negative L\'{e}vy process. In Section \ref{ss:mpg}, we consider an extension to more than two players. 

\subsection{Synchronization Game}
\label{ss:game}

Consider two players $A$ and $B$ who live in continuous time with infinite horizon. Both are active at time $0$ and can irreversibly terminate their activity (``exit'') at a time of their choice. In particular, at each time $t$, the players first sequentially pass through two decision nodes.\footnote{We introduce these discrete decision nodes at each continuous point in time $t$ to model the possibility that a player responds {\em immediately} to the exit  of the other player. This can be formalized as the limit of an appropriate discrete time game \citep{simon1989extensive}.}  

\vspace*{-10pt}\paragraph{Joint:} If neither player has exited in $[0,t)$, then they simultaneously decide on continuation. 

\vspace*{-10pt}\paragraph{Lone:} For $i=A,B$: If player $i$ is still active, but player $-i$ exited in $[0,t)$ or the time-$t$ Joint node, then player $i$ decides on continuation.\\
\vspace*{-10pt}

\noindent Next, each player $i$ that is still active earns utility $u^J_i(Y_t)$ if player $-i$ is still active and $u^L_i(Y_t)$ if player $-i$ exited in $[0,t]$. We assume that the utility flows are weakly decreasing in the external state $Y_t$, can take both negative and positive values, and are weakly consistent with strategic complementarities. 

\begin{assumption}[Monotonicity] \label{ass:M}
	For $S=J,L$; $i=A,B$: $u_i^S(y)$ weakly decreases in $y$, with $u_i^S(y)>0$ for low enough $y$ and $u_i^S(y)<0$ for high enough $y$.
\end{assumption}

\begin{assumption}[Complementarity] \label{ass:C}
For $i=A,B$; $u_i^J\left(y\right)\geq u_i^L\left(y\right) $ for all $y\in\R$.
\end{assumption}

\noindent Each inactive player earns outside payoffs $0$.

The external state $Y_t$ follows a spectrally negative L\'{e}vy process. A L\'{e}vy process $\{Y\}\equiv\{Y_t;t\geq 0\}$ is a stochastic process that starts at $Y_0=0$, with sample paths that are continuous from the right and have limits from the left,  and stationary and independent increments. It is the sum of a Brownian motion, with drift parameter $\mu\in\R$ and dispersion parameter $\sigma\in\Rp$, and an independent pure-jump process, characterized by its L\'{e}vy measure $\Pi$ \citep[for further details, see][Section S1]{ecsgsupp26:abbringyu}. The restriction to {\em spectrally negative} L\'{e}vy processes excludes positive jumps. This ensures that $\{Y\}$ attains and does not jump across a threshold when it first hits it from below, so that its hitting times can be analyzed like those of Brownian motion. Throughout, we maintain that $\{Y\}$ has a nontrivial Brownian motion component, $\sigma>0$.\footnote{This suffices to exclude the trivial case that $\{Y\}$ is never positive. It also implies a technical (ACP) condition of our optimal stopping analysis \citep[][Section S2.1]{ecsgsupp26:abbringyu}. Finally, it aids identification of some model specifications, along the lines of \citet{jem21:abbringsalimans}.} We do not parametrically restrict the L\'{e}vy measure $\Pi$ of $\{Y\}$, so our specification in terms of $(\mu,\sigma,\Pi)$ is a semiparametric extension of Section \ref{s:example}'s Brownian motion with drift. 

Both players have complete information. They value states and choices with their expected sum of utilities discounted at rates $\rho_A>0$ and $\rho_B>0$.\footnote{Throughout, we maintain that $\rho_A$ and $\rho_B$ are large enough to ensure convergence of expected discounted utilities \citep[Assumption S1]{ecsgsupp26:abbringyu}.} They employ pure Markov strategies that form a subgame-perfect (and therefore Markov-perfect) equilibrium. We will represent these strategies as pairs of node-specific Borel stopping sets  $(\mathcal{Y}^J_A,\mathcal{Y}^L_A)$ and $(\mathcal{Y}^J_B,\mathcal{Y}^L_B)$.

To solve for Markov perfect equilibrium, we need to find equilibrium strategies $(\mathcal{Y}^J_A,\mathcal{Y}^L_A)$ and $(\mathcal{Y}^J_B,\mathcal{Y}^L_B)$ such that each player's strategy solves the optimal stopping problem given the other player's strategy. As in Section \ref{s:example}'s simple example, our equilibrium analysis builds on the well-studied solutions of auxiliary optimal stopping problems and recurses backwards from Lone to Joint subgames. 

\subsection{Auxiliary Optimal Stopping Problems}
\label{ss:auxiliary}

We first characterize the optimal stopping strategies and the corresponding value functions in auxiliary optimal stopping problems in which players either stay and receive utility flow $u_{i}^{S}(Y_t)$, for either $S=J$ or $S=L$, or exit and receive $0$ forever. This serves two purposes. First,  the optimal stopping sets for the auxiliary problems with $S=L$ are $\mathcal{Y}_{A}^{L}$ and $\mathcal{Y}_{B}^{L}$, the equilibrium stopping sets for the Lone nodes. Second, their optimal values $\val_i^J(y)$ and $\val_i^L(y)$ of being active in state $y\in\R$ provide upper and lower bounds on the equilibrium values in Joint nodes. 

Our characterization builds closely on the literature on single-agent optimal stopping problems, notably the approach of \cite{boyarchenko_irreversible_2007} and, in particular, \cites{boyarchenko2005american} results for general L\'{e}vy processes. In \citet[Section S2.2]{ecsgsupp26:abbringyu}, we use this to show that, under Assumption \ref{ass:M},  the agent's optimal strategies are to exit as soon as $\{Y\}$ hits the exit set $[\overline{Y}_i^S,\infty)$.\footnote{The {\em open} exit set $(\overline{Y}_i^S,\infty)$ is also optimal. Because $\sigma>0$, its first hitting time is almost surely the same. As our focus is on implied durations, we will only present closed exit sets.} Here, the thresholds $\overline{Y}_i^S$; $S=J,L$; are uniquely determined by an intuitive condition. The details of this condition are not relevant to our equilibrium analysis, but we need some results on the optimal thresholds and on 
\begin{equation*} 
	\val _i^S\left( y;\overline{Y}\right) \equiv\mathbb{E}\left[ \int_{0}^{T\left(\overline{Y}-y\right)}\mathrm{e}^{-\rho_i t}u_i^S(y+Y_{t})dt\right],
\end{equation*}
the expected present value of a utility flow $u_i^S(y+Y_t)$ that is abandoned when $y+Y_t$ hits $\left[\overline{Y},\infty\right)$, at time $T\left(\overline{Y}-y\right)$. Note that the optimal value $\val_i^S(y)=\val_i^S(y;\overline{Y}_i^S)$. In \citet[][Section S2.2]{ecsgsupp26:abbringyu}, we prove\footnote{\label{fn:rho}See \citet[][Lemma S1]{ecsgsupp26:abbringyu}, which is Lemma \ref{lemma:aux} with a precise statement of the condition on $\rho_i$, Assumption S1.}
\begin{lemma}\label{lemma:aux}
If Assumption \ref{ass:M} holds, then (i) $\val^S_i\left( y;\overline{Y}\right) >0$ for $y<\overline{Y}\leq \overline{Y}^S_i$ and (ii) $\val^S_i\left( y;\overline{Y}\right) <0$ for $\overline{Y}^S_i<y<\overline{Y}$; $i=A,B$. If moreover Assumption \ref{ass:C} holds, then (iii) the optimal thresholds satisfy $\overline{Y}^{L}_i\leq \overline{Y}^{J}_i$; $i=A,B$.
\end{lemma}

\subsection{Equilibrium}
\label{ss:equilibrium}

The auxiliary problems with Lone payoffs $u_A^L$ and $u_B^L$ give the equilibrium strategies in player $A$ and $B$'s Lone subgames, $\mathcal{Y}_A^L=\left[ \overline{Y}_{A}^{L},\infty \right)$ and $\mathcal{Y}_B^L=\left[ \overline{Y}_{B}^{L},\infty \right)$. Given the corresponding Lone equilibrium values  $\brval_A^L\equiv \val_A^L$ and $\brval_B^L\equiv \val_B^L$, a stopping set $\mathcal{Y}_i^J$ is a best response to $\mathcal{Y}_{-i}^J$ if it attains player $i$'s best response value
\begin{equation*} 
\begin{split}
\brval^J_{i}\left( y\right) \equiv
\sup_{T\in\mathcal{T}}\mathbb{E}\Bigl[ 
	&\int_{0}^{\min\{T,T_{-i}^J(y)\}}\mathrm{e}^{-\rho_i t}u_{i}^{J}\left( y+Y_{t}\right) dt\\
	&~~~~~+\mathrm{e}^{-\rho_i \min\{T,T_{-i}^J(y)\}}\mathds{1}\left\{T>T^J_{-i}(y)\right\} \brval_{i}^{L}\left(y+ Y_{\min\{T,T_{-i}^J(y)\}}\right)
	\Bigr],
\end{split}
\end{equation*}
where $T^J_{-i}(y)\equiv\inf \left\{ t\geq 0\mid y+Y_{t}\in \mathcal{Y}^J_{-i}\right\}$; $i=A,B$; and $\mathds{1}\left\{\cdot\right\}=1$ if $\cdot$ is true and 0 otherwise. Here, $\mathcal{T}$ is the set of $\{Y\}$-stopping times. We first characterize these best responses. We will use that Assumption \ref{ass:C} (complementarity) implies that $\val_i^L\leq\brval_i^J\leq\val_i^J$; $i=A,B$.
\begin{lemma} \label{lemma:br} Under  Assumptions \ref{ass:M} and \ref{ass:C}, player $i$'s best response $\mathcal{Y}^J_i$ to any (equilibrium or out-of-equilibrium) $\mathcal{Y}^J_{-i}$ satisfies
(i) $\inf{\cal Y}^J_{i}\geq \overline{Y}_{i}^{L}$,
(ii) $\inf{\cal Y}^J_{i} \geq \inf{\cal Y}^J_{-i}$ if $\inf{\cal Y}^J_{-i}\leq\overline{Y}_{i}^{J}$, and
(iii) $\inf{\cal Y}^J_{i}=\overline{Y}_{i}^{J}$ if $\inf{\cal Y}^J_{-i} > \overline{Y}_i^J$;
$i=A,B$.
\end{lemma}
\begin{proof} (i) Using Assumption \ref{ass:M}, Lemma \ref{lemma:aux}(i) implies that  $\brval_i^J(y)\geq \val_{i}^{L}(y) > 0$, so that $y\not\in{\cal Y}^J_i$, for $y < \overline{Y}_{i}^{L}$. Therefore, $\inf{\cal Y}^J_i\geq \overline{Y}_{i}^{L}$.

(ii) In response to ${\cal Y}^J_{-i}$, player $i$ could pick ${\cal Y}_i^J=\left[\inf{\cal Y}^J_{-i},\infty\right)$ and earn $\val_{i}^{J}(y; \inf{\cal Y}^J_{-i})$. Player $i$'s best response will pay off no less, so $\brval^J_{i}(y)\geq \val_{i}^{J}(y; \inf{\cal Y}^J_{-i})$.  
If $\inf{\cal Y}^J_{-i}\leq \overline{Y}_{i}^{J}$ then, using Assumption \ref{ass:M}, Lemma \ref{lemma:aux}(i) implies that $\brval^J_{i}(y)\geq\val_{i}^{J}(y; \inf{\cal Y}^J_{-i}) > 0$, and therefore $y\not\in{\cal Y}_i^J$, for $y <\inf{\cal Y}^J_{-i}$. So, $\inf{\cal Y}^J_{i} \geq \inf{\cal Y}^J_{-i}$.

(iii) If $\inf{\cal Y}^J_{-i}>\overline{Y}_{i}^{J}$ then player $i$ could pick ${\cal Y}_i^J=\left[\overline{Y}^J_{i},\infty\right)$ and earn $\val_{i}^{J}$, so $\brval^J_{i}=\val_{i}^{J}$. Using Assumption \ref{ass:M}, Lemma \ref{lemma:aux}(i) implies that $\brval^J_{i}(y)=\val_{i}^{J}(y) > 0$, and therefore $y\not\in{\cal Y}_i^J$, for $y <\overline{Y}^J_{i}$. So, $\inf{\cal Y}^J_{i} \geq \overline{Y}^J_{i}$.

Now suppose that $\inf{\cal Y}^J_{i}>\overline{Y}^J_{i}$. Using Assumptions \ref{ass:M} and \ref{ass:C}, Lemma \ref{lemma:aux}(iii) implies that $\overline{Y} \equiv\inf\left({\cal Y}^J_{i}\cup{\cal Y}^J_{-i}\right)>\overline{Y}_{i}^{J}\geq \overline{Y}_{i}^{L}$, so that $\brval_i^J(y)=\val_i^J(y;\overline{Y})$ for $y<\overline{Y}$. Using Assumption \ref{ass:M}, Lemma \ref{lemma:aux}(ii) implies that $\brval^J_{i}(y)=\val_{i}^{J}(y; \overline{Y}) <0$ for $\overline{Y}^J_{i}<y<\overline{Y}$. So, ${\cal Y}^J_{i}$ with $\inf{\cal Y}^J_{i}>\overline{Y}^J_{i}$ is not a best response and $\inf{\cal Y}^J_{i} =\overline{Y}^J_{i}$.
\end{proof}

Equilibrium requires that $\mathcal{Y}^J_A$ and $\mathcal{Y}^J_B$ are best responses to each other in Joint subgames. As in Section \ref{s:example}'s  example, we distinguish two cases. 

\setcounter{case}{0}
\begin{case} {\bf Heterogeneity dominates complementarity}\label{case:sequentialgen}\\
If $\min \overline{Y}^J<\max \overline{Y}^L$, we can pick $i\in\{A,B\}$ such that $\overline{Y}_{i}^{L} \leq\overline{Y}_{i}^{J} < \overline{Y}_{-i}^{L} \leq \overline{Y}_{-i}^{J}$.
\begin{theorem}\label{th:eqSeq}
Suppose that  Assumptions \ref{ass:M} and \ref{ass:C} hold and that $\overline{Y}_{i}^{L} \leq\overline{Y}_{i}^{J} < \overline{Y}_{-i}^{L} \leq \overline{Y}_{-i}^{J}$. Then, (i)  $\inf{\cal Y}^J_{-i} \geq \overline{Y}_{-i}^{L}$ and (ii) $\inf{\cal Y}^J_{i} = \overline{Y}_{i}^{J}$ in any equilibrium. Moreover, (iii) all strategy profiles such that $\mathcal{Y}_{i}^{L} = [\overline{Y}_{i}^{L}, \infty)$, $\mathcal{Y}_{-i}^{L} = [\overline{Y}_{-i}^{L}, \infty)$, $\mathcal{Y}_{i}^{J} = [\overline{Y}_{i}^{J}, \infty)$, and $\mathcal{Y}_{-i}^{J} = [\overline{Y}_{-i}, \infty)$ for some $\overline{Y}_{-i} \geq \overline{Y}_{-i}^{L}$ form equilibria. 
\end{theorem}
\begin{proof}
(i) Lemma \ref{lemma:br}(i) implies $\inf{\cal Y}^J_{-i} \geq \overline{Y}_{-i}^{L}$. Combining this with $\overline{Y}_{-i}^{L}>\overline{Y}_{i}^{J}$, (ii) Lemma \ref{lemma:br}(iii) implies $\inf{\cal Y}^J_{i} = \overline{Y}_{i}^{J}$. 

(iii) We already established that $\mathcal{Y}_{i}^{L} = [\overline{Y}_{i}^{L}, \infty)$, $\mathcal{Y}_{-i}^{L} = [\overline{Y}_{-i}^{L}, \infty)$ are optimal in Lone subgames. 
In Joint nodes, player $i$ already earns the maximum value $\val_i^J$ by employing $\mathcal{Y}_{i}^{J} = [\overline{Y}_{i}^{J}, \infty)$, so cannot benefit from deviating. 

Player $-i$'s Joint exit behavior does not affect her actual exit behavior, and is trivially a best response, in states $y\geq \overline{Y}_{-i}^{L}>\overline{Y}_{i}^{J}$. 
With Assumption \ref{ass:C}, this also implies that she earns at least $\val_{-i}^L(y)\geq 0$, so cannot improve on $\mathcal{Y}_{-i}^{J}=[\overline{Y}_{-i}, \infty)$ by exiting in states $y\leq\overline{Y}_{-i}^{L}\leq \overline{Y}_{-i}$. Thus, $[\overline{Y}_{-i}, \infty)$ is a best response.
\end{proof}
\end{case}

The strategy profiles in Theorem \ref{th:eqSeq}(iii) are {\em all} threshold strategy profiles consistent with Theorem \ref{th:eqSeq}(i) and (ii). Thus, Theorem \ref{th:eqSeq}(iii) establishes that equilibria exist that meet Theorem \ref{th:eqSeq}(i) and (ii) for any choice of $\overline{Y}_{-i}\equiv\inf\mathcal{Y}_{-i}^J \geq \overline{Y}_{-i}^{L}$. It also provides a full characterization of the set of threshold equilibria in Case \ref{case:sequentialgen}.

The equilibrium multiplicity suggested by Theorem \ref{th:eqSeq}(i) and apparent in Theorem \ref{th:eqSeq}(iii) arises from the game's two-node structure: As noted in the proof, in states $y\geq \overline{Y}_{-i}^{L}$, player $-i$'s choices in a Joint node are inconsequential, because she will exit in the subsequent Lone node if she does not exit in the Joint node. Indeed, even though multiple (threshold and other) equilibria exist in this case, they all generate the same duration outcomes if $\min\overline{Y}^J\geq 0$ (which naturally arises if agents only enter profitable games): Player $i$ exits when $\{Y\}$ hits $\overline{Y}_{i}^{J}=\min \overline{Y}^J$ and player $-i$ leaves when $\{Y\}$ subsequently hits $\overline{Y}_{-i}^{L}=\max \overline{Y}^L>\min \overline{Y}^J$. Note that we are able to derive these empirical implications without further analyzing equilibria involving {\em disconnected} stopping sets, which would be complicated.\footnote{Unlike \emph{e.g.} \cite{murto_exit_2004}, in his analysis of a war of attrition driven by Brownian motion, we would have to account for the possibility that the state jumps down across (and not attain) the upper boundary of a stopping region in our L\'{e}vy case with disconnected stopping sets.}

\begin{case} {\bf Complementarity dominates heterogeneity}\label{case:synchronizationgen}\\
If $\min \overline{Y}^J\geq\max \overline{Y}^L$, we can pick $i\in\{A,B\}$ such that $\overline{Y}_{i}^{J} \leq \overline{Y}_{-i}^{J}$. To avoid unproductive distraction by the border case, we state our result for $\overline{Y}_{i}^{J}<\overline{Y}_{-i}^{J}$, so that either $\overline{Y}_{i}^{L} \leq \overline{Y}_{-i}^{L} \leq \overline{Y}_{i}^{J} < \overline{Y}_{-i}^{J}$ or $\overline{Y}_{-i}^{L} \leq \overline{Y}_{i}^{L} \leq \overline{Y}_{i}^{J} < \overline{Y}_{-i}^{J}$.\footnote{Theorem \ref{th:eqSim} carries over to the border case in which $\overline{Y}_{i}^{J}=\overline{Y}_{-i}^{J}$, but with an additional class of equilibria, a variant of Theorem \ref{th:eqSim}(ii) and (iv) in which $\inf\mathcal{Y}^J_{i} >\overline{Y}_{i}^{J}$ and $\inf\mathcal{Y}^J_{-i} = \overline{Y}_{i}^{J}$.}

\begin{theorem}\label{th:eqSim}
Suppose that  Assumptions \ref{ass:M} and \ref{ass:C} hold and that either $\overline{Y}_{i}^{L} \leq \overline{Y}_{-i}^{L} \leq \overline{Y}_{i}^{J} < \overline{Y}_{-i}^{J}$ or $\overline{Y}_{-i}^{L} \leq \overline{Y}_{i}^{L} \leq \overline{Y}_{i}^{J} <\overline{Y}_{-i}^{J}$. Then, either 
(i) $\inf\mathcal{Y}^J_{i} = \inf\mathcal{Y}^J_{-i} \in [\max\overline{Y}^{L},\overline{Y}_{i}^{J}]$ or 
(ii) $\inf\mathcal{Y}^J_{i} = \overline{Y}_{i}^{J}$ and $\inf\mathcal{Y}^J_{-i} >\overline{Y}_{i}^{J}$
in any equilibrium.  
Moreover, all strategy profiles such that $\mathcal{Y}_{i}^{L} = [\overline{Y}_{i}^{L}, \infty)$, $\mathcal{Y}_{-i}^{L} = [\overline{Y}_{-i}^{L}, \infty)$, and either 
(iii) $\mathcal{Y}_{i}^{J} = \mathcal{Y}_{-i}^{J} = [\overline{Y}, \infty)$ for some $\overline{Y} \in [\max\overline{Y}^{L}, \overline{Y}_{i}^{J}]$ or 
(iv) $\mathcal{Y}_{i}^{J} = [\overline{Y}_{i}^{J}, \infty)$ and $\mathcal{Y}_{-i}^{J} = [\overline{Y}_{-i}, \infty)$ for some $\overline{Y}_{-i} >\overline{Y}_{i}^{J}$
form equilibria. 
\end{theorem}
\begin{proof}
Lemma \ref{lemma:br}(i) implies that either (i) $\inf\mathcal{Y}^J_{-i}\in [\overline{Y}_{-i}^{L}, \overline{Y}_{i}^{J}]$ or (ii) $\inf\mathcal{Y}^J_{-i}>\overline{Y}_{i}^{J}$. 
(i) If $\inf\mathcal{Y}^J_{-i}\in [\overline{Y}_{-i}^{L}, \overline{Y}_{i}^{J}]$,  $\inf\mathcal{Y}^J_{i}\geq \inf\mathcal{Y}^J_{-i}$ by Lemma \ref{lemma:br}(ii). Conversely, $\inf\mathcal{Y}^J_{-i}\geq \inf\mathcal{Y}^J_{i}$ by Lemma \ref{lemma:br}(ii) and (iii) (as $\inf\mathcal{Y}^J_{-i}=\overline{Y}^J_{-i}>\overline{Y}^J_{i}$ is inconsistent with $\inf\mathcal{Y}^J_{-i}\in [\overline{Y}_{-i}^{L}, \overline{Y}_{i}^{J}]$). By Lemma \ref{lemma:br}(i), $\inf\mathcal{Y}^J_{i}\geq \overline{Y}_{i}^{L}$, so $\inf\mathcal{Y}^J_{i}=\inf\mathcal{Y}^J_{-i}\in [\max\overline{Y}^{L}, \overline{Y}_{i}^{J}]$.
(ii) If $\inf\mathcal{Y}^J_{-i}>\overline{Y}_{i}^{J}$, Lemma \ref{lemma:br}(iii) implies $\inf\mathcal{Y}^J_{i}=\overline{Y}^J_{i}$.

(iii) The  profile $(\mathcal{Y}_{i}^{J},\mathcal{Y}_{-i}^{J}) = ([\overline{Y}, \infty),[\overline{Y}, \infty))$ gives values $\val^J_i(y;\overline{Y})$ and $\val_{-i}^J(y;\overline{Y})$. Because $\overline{Y}\leq \overline{Y}_i^J\leq \overline{Y}_{-i}^J$,  Lemma \ref{lemma:aux}(i) implies that $\val^J_i(y,\overline{Y}_{i}^{J})>0$ and $\val^J_{-i}(y,\overline{Y}_{-i}^{J})>0$, so that neither player profits from exiting instead, at $y<\overline Y$. In states $y\geq\overline Y$, each player's Joint exit behavior is again inconsequential and thus trivially best. 

(iv) If $\mathcal{Y}_{-i}^{J} = [\overline{Y}_{-i}, \infty)$ for some $\overline{Y}_{-i} >\overline{Y}_{i}^{J}$, player $i$ already earns the maximum value $\val_i^J$ by employing $\mathcal{Y}_{i}^{J} = [\overline{Y}_{i}^{J}, \infty)$, so cannot benefit from deviating. 

Player $-i$'s Joint exit behavior does not affect her actual exit behavior, and is trivially a best response, in states $y\geq \overline{Y}_{i}^{J}$. In states $y<\overline{Y}_{i}^{J}\leq\overline{Y}_{-i}^{J}$, player $-i$  continues with value $\val^J_{-i}(y,\overline{Y}_{i}^{J})$, $\val^J_{-i}(y,\overline{Y}_{i}^{J})>0$ by Lemma \ref{lemma:aux}(i), and player $-i$ cannot profit from exiting instead. So, $[\overline{Y}_{-i}, \infty)$ is a best response.
\end{proof}

The strategy profiles in Theorem \ref{th:eqSim}(iii) and (iv) are {\em all} threshold strategy profiles consistent with, respectively, Theorem \ref{th:eqSim}(i) and (ii). Theorem \ref{th:eqSim}(iii) provides examples of equilibria consistent with Theorem \ref{th:eqSim}(i) for any choice of $\overline{Y}\equiv\inf\mathcal{Y}^J_{i} = \inf\mathcal{Y}^J_{-i}\in[\max\overline{Y}^{L}, \overline{Y}_{i}^{J}]$. Theorem  \ref{th:eqSim}(iv) similarly shows that equilibria satisfying Theorem \ref{th:eqSim}(ii) exist for any $\overline{Y}_{-i}\equiv\inf\mathcal{Y}^J_{-i}>\overline{Y}_{i}^{J}$. Together, they fully characterize the set of threshold equilibria in Case \ref{case:synchronizationgen}.%

The equilibrium multiplicity in Theorem \ref{th:eqSim}(ii) and (iv) is similar to that in Case \ref{case:sequentialgen}. It arises from the game's two-node structure and gives unique duration outcomes if $\min\overline{Y}^J\geq 0$: Both players exit when $\{Y\}$ hits $\overline{Y}_i^J=\min\overline{Y}^J$.
The multiplicity in Theorem \ref{th:eqSim}(i) and (iii), however, translates into multiple duration outcomes: If $\overline{Y}\geq 0$, both players exit when $\{Y\}$ hits $\overline{Y} \in [\max\overline{Y}^{L}, \overline{Y}_{i}^{J}]$. Without further restrictions, an empirical implementation of our game for survival data will not fully determine the distribution of the observed durations and thus be {\em incomplete} \citep{tamer_incomplete_2003}. We could leverage Theorems \ref{th:eqSeq} and \ref{th:eqSim} to implement any of the literature's approaches to such incompleteness. We will rely on a carefully motivated selection of a, for expositional convenience, threshold equilibrium.\footnote{Alternatively, we could remain agnostic about equilibrium selection and study (typically) partial identification and set estimation of the model. See \citet{tamer_incomplete_2003}.}

To this end, note that player's $i$ choice of $\mathcal{Y}_{i}^{J} = [\overline{Y}_{i}^{J}, \infty)$ weakly dominates $[ \overline{Y}, \infty)$ for $\overline{Y}<\overline{Y}_{i}^{J}$. So, by eliminating equilibria involving weakly dominated strategies, we select the unique equilibrium such that $\mathcal{Y}_{i}^{J} = [\overline{Y}_{i}^{J},\infty)$. This refinement aligns with imposing trembling-hand perfection.\footnote{A formal trembling-hand implementation would posit that player $-i$ might mistakenly exit at $\inf\mathcal{Y}^J_{-i} + \epsilon_{-i}$, where $\epsilon_{-i}$ is an $\R$-valued random perturbation unknown to both players. Under this perturbation, $\inf\mathcal{Y}^J_{i} = \overline{Y}_{i}^{J}$ emerges as player $i$'s unique best response according to Lemma \ref{lemma:br}.} 
\end{case}

Our econometric analysis will assume this equilibrium refinement, so that the players employ $\mathcal{Y}_{i}^{S} = [\overline{Y}_{i}^{S},\infty)$; $S = J,L$; $i = A,B$. This completely specifies the duration outcomes. 
If the game's primitives are such that $\min \overline{Y}^{J} < \max \overline{Y}^{L}$ (Case \ref{case:sequentialgen}), the player with the lower thresholds exits when $\{Y\}$ hits $\min\overline{Y}^J$ and the other player when $\{Y\}$ subsequently hits $\max\overline{Y}^L$. We will refer to this outcome as ``sequential exit''. 
If the primitives imply $\min \overline{Y}^{J}\geq \max\overline{Y}^{L}$ (Case \ref{case:synchronizationgen}), both players exit when $\{Y\}$ hits $\min \overline{Y}^{J}$. We will refer to this as  ``simultaneous exit''.

\subsection{Extension to More than Two Players} \label{ss:mpg}

We now extend the two-player synchronization game to general sets ${\cal I}$ of $\np\geq 2$ agents. The game structure preserves sequential decision nodes indexed by the current set of  survivors $\mathcal{S} \subseteq \mathcal{I}$.
This embeds the two-player case as $\mathcal{I} = \{A,B\}$, with $\mathcal{S} = \{A,B\}$ for the Joint node and $\mathcal{S} = \{A\}$ or $\{B\}$ for the Lone node. At each instant $t$, survivors sequentially pass through decision nodes. In each node, players simultaneously decide whether to continue or exit permanently based on complete information. When a set $\mathcal{E}$ of players exits in node $\mathcal{S}$, the game transitions to node $\mathcal{S} \setminus \mathcal{E}$. Time continues only after no exits occur in the most recent decision node. Before that, agent $i\in{\cal I}$ receives her flow payoff at time $t$:  $u^{|\mathcal{S}|}_i(Y_t)$ if she is among $\mathcal{S}$ survivors or $0$ once she has exited.\footnote{This specification implicitly assumes homogeneous complementarity: A survivor's payoffs only depend on the number of surviving players, not their identities. Heterogeneity would introduce additional notation without altering the core economic insights.} 

We again assume that utility is monotonic and embodies complementarity.
\begin{assumption} \label{ass:MM}
For $S = 1, 2, \dots, \np$ and $i \in \mathcal{I}$: $u_i^S$ is weakly decreasing in $y$, with $u_i^S(y) > 0$ for sufficiently low $y$ and $u_i^S(y) < 0$ for sufficiently high $y$. Moreover, $u_i^S$ is weakly decreasing in $S$: $u_i^{\np} \geq \cdots \geq u_i^2 \geq u_i^1$.
\end{assumption}
\noindent Assumption \ref{ass:MM} ensures that the solution to the auxiliary optimal stopping problem with utility $u_i^S$ is a threshold rule. As in the two-player case, the optimal thresholds $\overline{Y}_i^S$ and values $\val_i^S$ derived from such problems play a central role in the equilibrium characterization. Using that $u_i^S$ weakly decreases in $S$ (Assumption \ref{ass:MM}), Lemma \ref{lemma:aux}(iii) establishes that $\overline{Y}_i^{\np} \geq \overline{Y}_i^{\np-1} \geq \cdots \geq \overline{Y}_i^2 \geq \overline{Y}_i^1$ for each player $i$.
 
Markov-perfect equilibrium is defined in terms of stopping sets $\{ \mathcal{Y}_i^{\mathcal{S}} \}_{i \in \mathcal{S} \subseteq \mathcal{I}}$.\footnote{Here, ${i \in \mathcal{S} \subseteq \mathcal{I}}$ is shorthand for all pairs $(i,{\cal S})$ such that $i \in \mathcal{S}$ and $\mathcal{S}\subseteq \mathcal{I}$.} As the number of exit scenarios grows factorially with $\np$, considering all possible equilibria is prohibitive.\footnote{In \citet[][Section S3]{ecsgsupp26:abbringyu}, we extend some of Section \ref{ss:equilibrium}'s equilibrium analysis to general ${\cal I}$, but we do not use this here.} We therefore apply Section \ref{ss:equilibrium}'s refinement to select the unique equilibrium with $\mathcal{Y}_{i}^{\mathcal{S}} = [\overline{Y}_{i}^{|\mathcal{S}|}, \infty)$.
\begin{theorem}\label{th:eqGen}
The strategy profile $\{\mathcal{Y}_{i}^{\mathcal{S}} = [\overline{Y}_{i}^{|\mathcal{S}|}, \infty)\}_{i\in \mathcal{S} \subseteq \mathcal{I}}$ forms an equilibrium. 
\end{theorem}
\begin{proof}
	If $\overline{Y}_{i}^{|\mathcal{S}|}=\min_{j\in \mathcal{S}} \overline{Y}_{j}^{|\mathcal{S}|}$, player $i$ already earns the maximum value $\val_i^{|\mathcal{S}|}$ by employing $\mathcal{Y}_{i}^{\mathcal{S}} = [\overline{Y}_{i}^{|\mathcal{S}|}, \infty)$, so cannot benefit from deviating. If $\overline{Y}_{i}^{|\mathcal{S}|}>\min\overline{Y}_{-i}^{|\mathcal{S}|}\equiv \min_{j\in \mathcal{S}\setminus i} \overline{Y}_{j}^{|\mathcal{S}|}$, player $i$'s actual exit behavior does not change with the strategy in states $y\geq\min\overline{Y}_{-i}^{|\mathcal{S}|}$, so is
	trivially a best response. In states $y<\min\overline{Y}_{-i}^{|\mathcal{S}|}<\overline{Y}_{i}^{|\mathcal{S}|}$, player $i$ continues with at least value $\val_i^{|\mathcal{S}|}(y,\min\overline{Y}_{-i}^{|\mathcal{S}|})$, which is positive by Lemma \ref{lemma:aux}(i). She cannot profit from exiting instead. So, $\{[\overline{Y}_{i}^{|\mathcal{S}|}, \infty)\}_{i\in \mathcal{S} \subseteq \mathcal{I}}$ are mutual best responses.
\end{proof}

This refinement yields equilibrium exit dynamics characterized by $W\leq \np$ distinct waves indexed by $w$, where each wave features synchronized exits triggered when $\{Y\}$ hits an increasing sequence of thresholds.

The game starts with survivors ${\cal S}^1 = {\cal I}$ before the first exit wave ($w=1$). Wave $w$ begins when $\{Y\}$ hits the wave-trigger threshold $\barYw^w\equiv\min_{i\in{\cal S}^w}\overline{Y}^{|{\cal S}^w|}_i$ (this may never happen, in which case wave $w$ is not realized; we return to this in Section \ref{ss:spec}). At this threshold, all players $j$ with individual thresholds matching $\barYw^w$ immediately exit, forming the wave's first exit cohort ${\cal E}^w_1 \equiv \{j\in{\cal S}^w \mid \overline{Y}^{|{\cal S}^w|}_j =\barYw^w\}$. 

Following these initial exits, the game transitions to subsequent decision nodes within the same wave. In node $k \geq 2$, additional players exit if their continuation thresholds satisfy $\overline{Y}^{|{\cal S}^w \setminus \cup_{l=1}^{k-1} {\cal E}^w_l|}_j \leq \barYw^w$, forming cohort ${\cal E}^w_k \equiv \{j\in{\cal S}^w\setminus\cup_{l=1}^{k-1}{\cal E}^w_l \mid \overline{Y}^{|{\cal S}^w \setminus \cup_{l=1}^{k-1} {\cal E}^w_l|}_j \leq \barYw^w\}$. This cascading process continues until no further players exit at $\barYw^w$, with ${\cal E}^w_k = \emptyset$ implying termination of intra-wave exits. The full exit set for wave $w$ is ${\cal E}^w \equiv \cup_k{\cal E}^w_k$.  Crucially, all exits in wave $w$ occur simultaneously when $\{Y\}$ first hits $\barYw^w$. 

Wave $w's$ survivors ${\cal S}^{w+1} \equiv {\cal S}^w \setminus {\cal E}^w$ proceed to wave $w+1$ if ${\cal S}^{w+1} \neq \emptyset$. If there are no survivors (${\cal S}^{w+1}= \emptyset$), the game ends with $W=w$.

\section{Identification}
\label{s:identification}

Consider learning about the game's primitives from data on its duration outcomes in case we observe some but not all of the players' characteristics. Sections \ref{ss:spec}--\ref{ss:iden2} focus on a ``single-game'' setup,  in which we only observe the outcome of one game for each draw of the unobserved characteristics. Section \ref{ss:extensions} considers alternative setups, including a ``grouped-game'' setting, in which we observe the outcomes of multiple games for each draw of the unobserved characteristics.

\subsection{Specification and Characterization} 
\label{ss:spec}

In the general case, denote the set of players with ${\cal I}=\left\{1,\ldots,\np\right\}$. 
Suppose we randomly sample from (and therefore can identify) the distribution of $(\mathbf{T},\mathbf{X})$, where $\mathbf{T}\equiv(T_1,\ldots.T_\np)$ collects the exit times and $\vX\equiv(\X_1,\ldots,X_\np)$ the observed characteristics of the players.\footnote{For expositional simplicity, we assume a fixed number $\np$ of players across games. With the structure on payoffs that we will assume, variation in $\np$ across games could aid identification.} Player $i$'s preferences are fully determined by $\X_i$ and an unobserved scalar random variable $\err_i$; in the sense that $u_i^1(\cdot)=\bar u^1(\cdot;\X_i,\err_i),\ldots,u_i^\np(\cdot)=\bar u^\np(\cdot;\X_i,\err_i),$ and $\rho_i=\bar\rho(\X_i,\err_i)$ for some measurable functions $\bar u^1,\ldots,\bar u^\np$, and $\bar\rho$; $i\in{\cal I}$. The observed characteristics $\vX$ are randomly drawn from some distribution on $\sX^\np$, where $\sX\subseteq\R^K$ is the support of $\X_i$; $\bigerr\equiv(\err_1,\ldots,\err_\np)$ is independent of $\vX$ with distribution $G$ on $\Rp^\np$; and the latent process $\{Y\}$ is independent of $(\vX,\bigerr)$.  The observed exit times $\mathbf{T}$ are Theorem \ref{th:eqGen}'s equilibrium outcomes of the game with characteristics $(\vX,\bigerr)$. The implied distribution of $\mathbf{T}\mid\vX,\bigerr$ depends on 
\begin{enumerate}
\item\label{enum:Y}  $(\mu,\sigma,\Pi)$ through $\{Y\}$ and 
\item\label{enum:Ybar} $(\mu,\sigma,\Pi)$, $\bar u^1,\ldots,\bar u^\np$, $\bar\rho$, and $(\vX,\bigerr)$ through $\overline Y^1_i,\ldots,\overline Y^\np_i$; $i\in{\cal I}$. 
\end{enumerate}

To make \ref{enum:Y} precise, note that the distributions of $\{Y\}$ are fully characterized by its {\em Laplace exponent} $\psi(\s)\equiv \ln\mathbb{E}\left[\exp\left(\s Y_1\right)\right]$ for $\s\geq 0$. The L\'{e}vy-Khintchine formula provides an explicit parameterization of $\psi$ in terms of $(\mu,\sigma,\Pi)$. Because $\psi$ is continuous and convex, with $\psi (0)=0$ and $\lim_{\s\rightarrow \infty }\psi(\s)=\infty$, $\psi(\s)=0$ has one or two solutions. Denote the largest of these with $\s_0\geq 0$. Then, the inverse $\Lambda:\Rp\rightarrow [\s_{0},\infty)$ of $\psi:[\s_0,\infty)\rightarrow\Rp$ is the {\em Laplace exponent} of the hitting-time process $\{T\}\equiv\{T(y);y\geq 0\}$. Through the Laplace transform
\begin{equation}
\label{eq:lapTy}
\s \in [0, \infty)\mapsto\Lap_{T(y)}(\s) \equiv \mathbb{E}\left[ \exp(-\s T(y)) \mathds{1}\left\{T(y)<\infty\right\} \right]=\exp\left(-\Lambda(\s)y\right)
\end{equation}
of $T(y)$, it fully characterizes the distribution of $T(y)$ \citep[][Section XIII.1]{feller_introduction_1971}. If $\psi'(0)\geq 0$, then $\s_0=0$, $\psi:\Rp\rightarrow\Rp$ is strictly increasing, $\Lambda$ is simply its inverse, and $\Pr\left(T(y)<\infty\right)=\Lap_{T(y)}(0)=1$. In particular, if $\{Y\}$ is a Brownian motion ($\Pi=0$) with nonnegative drift ($\mu\geq 0$), as in Section \ref{s:example}'s example, then $\psi(\s)=\mu \s+\sigma^2\s^2/2$ and $\Lambda(\s)=\left(-\mu+\sqrt{\mu^2+2\s\sigma^2}\right)/\sigma^2$. If $\psi'(0)<0$, then $\s_0>0$ and $\Pr\left(T(y)<\infty\right)=\exp\left(-\s_0 y\right)<1$. \citet[][Section S1]{ecsgsupp26:abbringyu} provides details.

To illustrate \ref{enum:Ybar}, consider an $\np$-player version of Section \ref{s:example}'s example with general $\{Y\}$, $R^S_i = R^S(\X_i)\exp(\err_i)$, $C^S_i = C^S(\X_i)$, $\rho_i=\rho>\psi(1)$, and $\gamma_i = 1$. The flow payoffs then simplify to $u_i^S(Y_{t}) = R^S(\X_i)\exp(\err_i) - C^S(\X_i)\exp(Y_t)$. The corresponding exit thresholds are \citep[][Section S2.3]{ecsgsupp26:abbringyu}
\begin{equation}
\overline{Y}_{i}^{S}=\phi ^{S}\left( \X_{i}\right) +\err_{i},  \label{eq:addsep}
\end{equation}%
\noindent where $\phi^{S}(\X_i) \equiv \ln\left( \frac{R^S(\X_i)}{C^S(\X_i)} \cdot \frac{\Lambda(\rho)}{\Lambda(\rho)-1} \cdot \frac{\rho-\psi(1)}{\rho} \right)$. For our main analyses, we adopt this additively separable specification of the thresholds for general measurable $\phi^S$, so without assuming it has the example's structure (in Section \ref{ss:extensions}, we discuss multiplicatively separable and nonseparable alternatives). Throughout, we will focus on identifying $\phi^1,\ldots,\phi^{\np}$. Subsequently determining $\bar u^1,\ldots,\bar u^{\np}$, and $\bar\rho$ is nontrivial, but standard.  Even in our tightly specified example, we cannot determine $R^S$ and $C^S$ from $\phi^S$ without further (exclusion) restrictions. Generally, we can address this with results for {\em single-agent} models as in \citet[][Section 4.3]{are10:abbring}, because our game's thresholds solve auxiliary optimal stopping problems. 

Given \ref{enum:Y} a specification of $\{Y\}$ and \ref{enum:Ybar} the thresholds in \eqref{eq:addsep}, we can characterize the distribution of $\mathbf{T}\mid\vX,\bigerr$. To this end, we represent $\mathbf{T}$ by the $W^*$ distinct exit times $T^1<\cdots<T^{W^*}$ in $\mathbf{T}$ and the identities $\left({\cal E}^1,\ldots,{\cal E}^{W^*-1}\right)$ of the players exiting at the first $W^*-1$ of those times. If $T^{W^*}<\infty$, the  players remaining after the first $W^*-1$ waves, those in ${\cal S}^{W^*}={\cal I}\setminus\bigcup_{w=1}^{W^*-1}{\cal E}^w$, exit at $T^{W^*}$: ${\cal E}^{W^*}={\cal S}^{W^*}$. In this case, all exit waves are realized and observed ($W^*=W$). If  $T^{W^*}=\infty$, the $W^*$-th and any later waves never materialize and $W^*\leq W$. 

Consider the game for a fixed value $\vx\in\sX^\np$ of $\vX$ and $\bigerr$ drawn from $G$. This game's auxiliary thresholds $\phi ^{S}\left( \x_{i}\right) +\err_{i}$ and thus its {\em exit pattern}--- the number of (potential) exit waves $W$ and their compositions $({\cal E}^1,\ldots,{\cal E}^W)$--- are fully determined by $\bigerr$. The exit times $T^1,\ldots,T^{W}$ are the first times $\{Y\}$ hits the increasing sequence of thresholds $\barYw^w(\vx)\equiv\min_{i\in{\cal S}^w}\left\{\phi ^{|{\cal S}^w|}\left( \x_{i}\right) +\err_{i}\right\}$; $w=1,\ldots,W$. Because $\{Y\}$ is a strong Markov process, given $\bigerr$ and on $\left\{T^{w-1}<\infty\right\}$, the inter-wave duration $\Delta T^w\equiv T^w-T^{w-1}$ is independently distributed from $(T^1,\ldots,T^{w-1})$ as the hitting time $T\left(\Delta\barYw^w(\vx)\right)$, where $\Delta\barYw^w(\vx)\equiv\barYw^w(\vx)-\barYw^{w-1}(\vx)>0$ (here, $T^0\equiv 0$ and $\barYw^0(\vx)\equiv 0$). So, using \eqref{eq:lapTy} and for $w\leq W$, the Laplace transform of the distribution of $\left(\Delta T^1,\ldots,\Delta T^{w}\right)\mid\bigerr$ for the game with observed characteristics $\vx$ on $\bigcap_{j=1}^w\left\{\Delta T^j<\infty\right\}=\left\{T^w<\infty\right\}$ is 
\begin{equation*}
\mathbf{\s}^w\mapsto\mathbb{E}\left[ \exp\left(-\sum_{j=1}^{w}\s^j\Delta T^j\right) 
\mathds{1}\left\{T^w<\infty\right\} \;\vline\;\bigerr\right]
=
\exp\left(-\sum_{j=1}^{w}\Lambda(\s^j)\Delta\barYw^j(\vx)\right),
\end{equation*}
\noindent where $\mathbf{\s}^w\equiv(\s^1,\ldots,\s^w)\in\Rp^w$. As ${\cal E}^1,\ldots,{\cal E}^w$ are nonempty and disjoint subsets of ${\cal I}$ that may not cover ${\cal I}$, they form an {\em ordered partial partition} of ${\cal I}$. Let $E^w(\mathbf{e}^w;\vx)\subseteq\Rp^\np$ denote the set of $\bigerr$ such that ${\cal E}^1=e^1,\ldots,{\cal E}^w=e^w$ for such an ordered partial partition $\mathbf{e}^w\equiv\left(e^1,\ldots,e^w\right)$ of ${\cal I}$. Then, the Laplace transform of the same distribution on $\left\{T^w<\infty\right\}\bigcap\left\{{\cal E}^1=e^1,\ldots,{\cal E}^w=e^w\right\}$ is
$\mathbf{\s}^w\mapsto\exp\left(-\sum_{j=1}^{w}\Lambda(\s^j)\Delta\barYw^j(\vx)\right)\mathds{1}_{E^w(\mathbf{e}^w;\vx)}(\bigerr)$, where $\mathds{1}_{A}(\cdot)$ is short for $\mathds{1}\left\{\cdot\in A\right\}$. Integrating over the distribution $G$ of $\bigerr$ gives the corresponding unconditional Laplace transform
\begin{equation}
\label{eq:LapX}
\Lap^w\left(\mathbf{\s}^w;\mathbf{e}^w \mid \vx\right) 
=\mathbb{E}\left[\exp\left(-\sum_{j=1}^{w}\Lambda(\s^j)\Delta\barYw^j(\vx)\right)\mathds{1}_{E^w(\mathbf{e}^w;\vx)}(\bigerr)\right], 
\end{equation}

\noindent which uniquely determines the distribution of $\left(\Delta T^1,\ldots,\Delta T^{w}\right)$, and thus $\left(T^1,\ldots,T^{w}\right)$, on $\left\{T^w<\infty\right\}\bigcap\left\{{\cal E}^1=e^1,\ldots,{\cal E}^w=e^w\right\}$ for the game with characteristics $\vx$. 
In Appendix \ref{app:characterization}, we prove the following characterization result.
\begin{lemma}\label{lemma:characterization}
The distribution of $\mathbf{T}$ for given $\vx$ is one-to-one related to $\Lap^w\left(\cdot\; ;\mathbf{e}^w \mid \vx\right)$ for all ordered partial partitions $\mathbf{e}^w$ of ${\cal I}$; $w=1,\ldots,\np$. 
\end{lemma}

The data only determine the distribution of $\mathbf{T}\mid\vX$ almost surely. We resolve this ambiguity by assuming continuity of the thresholds in the covariates on their support and absolute continuity of $G$.
\begin{assumption}
\label{ass:cont}
Let $B(\x,\epsin)$ be an open ball of radius $\epsin>0$ around $\x\in\R^K$. For each $\x\in\sX$, $\lim_{\epsin\downarrow 0}\sup_{\x'\in B(\x,\epsin)\cap\sX}|\phi^S(\x')-\phi^S(\x)|=0$; $S=1,\ldots,\np$.
\end{assumption}
\begin{assumption} \label{ass:G}
	$G$ is absolutely continuous with density $g > 0$ on $\Rp^\np$.
\end{assumption}
\noindent Assumptions \ref{ass:cont} and \ref{ass:G} ensure that $\Lap^w\left(\cdot\; ;\cdot\mid\vx\right)$ is continuous in $\vx$, so that it is uniquely determined for all $\vx\in\sX^\np$ from the distribution of $(\mathbf{T},\vX)$.\footnote{See \citet[][Lemma S3]{ecsgsupp26:abbringyu}. This result does not use that $g>0$. This additional assumption ensures statistical nondegeneracy and will be used later.} Thus, we will say that some or all of the model primitives $\left( \mu, \sigma, \Pi, \phi^1,\ldots,\phi^{\np}, G \right)$ are {\em identified} if they are uniquely determined from $\Lap^w\left(\cdot\; ;\mathbf{e}^w \mid \vx\right)$ for all $\vx\in\sX^\np$ and ordered partial partitions $\mathbf{e}^w$ of ${\cal I}$; $w=1,\ldots,\np$. Because the parameters $\left( \mu, \sigma, \Pi \right)$ of the latent L\'{e}vy process are one-to-one related to $\Lambda$ \citep[][Section S1.2]{ecsgsupp26:abbringyu}, we can and will alternatively focus on identifying $\left( \Lambda, \phi^1,\ldots,\phi^{\np}, G \right)$.

To appreciate the challenges in identifying our game, consider what we can learn from the distribution of $T^1\mid\vX$, which we can characterize by the Laplace transform $\Lap_{\Delta T^1}\left(\cdot\mid\vx\right)=\sum_{e^1\subseteq{\cal I},e^1\neq\emptyset}\Lap^1\left(\cdot\; ;e^1\mid\vx\right)$, $\vx\in\sX^\np$. Substituting \eqref{eq:LapX} gives 
\begin{equation}
\label{eq:LapOne}
\Lap_{\Delta T^1}\left(\s\mid\vx\right)=\mathbb{E}\left[\exp\left(-\Lambda(\s)\min_{i\in{\cal I}}\left\{\phi^{\np}\left( \x_{i}\right) +\err_{i}\right\}\right)\right].
\end{equation}

\noindent In the single-agent case ${\cal I}=\{i\}$, the (only) exit time $T^1=T\left[\phi^1(\X_i)+\err_i\right]$. Then, \eqref{eq:LapOne} would express the data, represented by $\Lap_{\Delta T^1}\left(\cdot\mid\x_i\right)$, in terms of $\Lambda$, which characterizes $\{Y\}$, and agent $i$'s threshold $\phi^1(\x_i)+\err_i$. This case can be handled by a variant, for additive separable thresholds, of the identification analysis in \cite{abbring_mixed_2012}. Identification in our game with $\np\geq 2$ players is complicated by the fact that it only indirectly relates $T^1$ to the individual thresholds, through $\Delta\barYw^1=\min_{i\in{\cal I}}\left\{\phi^{\np}\left( \x_{i}\right) +\err_{i}\right\}$. Moreover, we only observe exit behavior in states with $S<\np$ remaining players, which is needed to identify $\phi^S$, selectively for those games with exit patterns that imply a wave $w$ with $|{\cal S}^w|=S$ survivors and that actually experience that wave, $T^{w-1}<\infty$. The next section addresses these challenges for two-player games. Section \ref{ss:iden2} returns to the general case.

\subsection{Identification of Two-Player Games} \label{ecsg:s:iden}

In the two-player case, we continue to denote $\mathcal{I}=\{A,B\}$ (instead of $\left\{1,2\right\}$) and refer to the Lone and Joint states $S$ as $L$ and $J$ (instead of 1 and 2). We analyze the identification of $(\Lambda, \phi^L, \phi^J, G)$ from $\Lap^w\left(\cdot;\mathbf{e}^w\mid\vx\right)$ for all ordered partial partitions $\mathbf{e}^w$ of $\{A,B\}$ and $\vx=(x_A,x_B)\in\sX^2$; $w=1,2$; in two steps. 

First, consider determining $\phi^J$ and $\Lambda$ from the distribution of $T^1\mid\vX$, irrespective of whether exit is simultaneous (${\cal E}^1=\{A,B\}$) or sequential (${\cal E}^1=\{A\}$ or ${\cal E}^1=\{B\}$). Recall that $T^1$ is the first time $\{Y\}$ hits $\overline{Y}^J=\min\left\{\phi^J(\X_A)+\err_A,\phi^J(\X_B)+\err_B\right\}$. The distributions of $T^1\mid\vx$, $\vx\in\sX^2$, are characterized by a special case of the Laplace transform in \eqref{eq:LapOne} for $I=2$: 
\begin{align*}
\Lap_{\Delta T^1}\left(\s\mid\vx\right)
&=\sum_{e^1=\{A,B\},\{A\},\{B\}}\Lap^1\left(\s;e^1\mid\vx\right)\\
&=
\mathbb{E}\left[\exp\left(-\Lambda(\s)\min\left\{\phi^J(\x_A)+\err_A,\phi^J(\x_B)+\err_B\right\}\right)\right].
\end{align*}

\noindent If both players have the same covariate values $\x_A=\x_B=\x$, this reduces to
\begin{equation} \label{eq:lap1}
	\Lap_{\Delta T^{1}}(\s \mid (\x,\x)) = \exp\left[ -\phi^{J}(\x) \Lambda(\s) \right] 
	\Lap_{\hat{\err}}\left( \Lambda(\s) \right),
\end{equation}
\noindent where $\Lap_{\hat{\err}}$ is the Laplace transform of $\hat\err\equiv\min\left\{\err_A,\err_B\right\}$. We can separate the latent process and the unobserved heterogeneity if we have some variation in $\phi^J(x)$.
\begin{assumption} \label{ass:variation}
	$\phi^{J}(x)$ takes at least two distinct values for $x\in\sX$.
\end{assumption}
\begin{lemma} \label{lemma:identTwo}
	Under Assumptions \ref{ass:cont}--\ref{ass:variation}, $c\phi ^{J}$, $c^{-1}\Lambda$, and the distribution of $c\hat\err$ are identified up to a scale factor $c\in\Rpp$.
\end{lemma}
\begin{proof}
	By Assumption~\ref{ass:variation}, we can find $\x',\x''\in\sX$ such that $\phi^{J}(\x') \ne \phi^{J}(\x'')$ by checking its implication $\Lap_{\Delta T^{1}}(\s \mid (\x',\x')) \ne \Lap_{\Delta T^{1}}(\s \mid (\x'',\x''))$. Evaluating \eqref{eq:lap1} at these values and subsequently taking ratios and logarithms gives
	\begin{equation*} 
		\ln \frac{\Lap_{\Delta T^{1}}(\s \mid (\x',\x'))}{\Lap_{\Delta T^{1}}(\s \mid (\x'',\x''))} 
		= \left[ \phi^{J}(\x'') - \phi^{J}(\x') \right] \Lambda(\s), 
	\end{equation*}
	\noindent which identifies $c^{-1}\Lambda(\s)$, $\s\in\Rp$, up to a scale factor $c\in\Rpp$. Using this, \eqref{eq:lap1} identifies $z\in\Rp\mapsto c\phi^{J}(\x)\s-\ln\Lap_{\hat{\err}}\left(c \s\right)$, and thus $c\phi^{J}(\x)+d$, $x\in{\cal X}$, up to the same $c\in\Rpp$ and a location shift $d\in\R$. Now, suppose we start with $\phi^J$, $\Lambda$, $\Lap_{\hat\err}$ satisfying \eqref{eq:lap1} and Assumption \ref{ass:G}. Then, if we add $d$ to $\phi^J$, we need to substract $d$ from $\hat\err$ (multiply its Laplace transform by $\exp(-d\s)$) to satisfy \eqref{eq:lap1}. But this changes the support of $\hat\err$ to $[-d,\infty)$, which violates Assumption \ref{ass:G} unless $d=0$. So, $c\phi^J$ and $z\mapsto\Lap_{\hat\err}(c z)=\Lap_{c\hat\err}(z)$ are identified up to $c\in\Rpp$.
\end{proof}

\noindent Identification requires one scale normalization, because \eqref{eq:lap1} and Assumption \ref{ass:G} still hold if we multiply both the effective threshold $\phi^J(\x)+\hat\err$ and the latent process $\{Y\}$ by a factor $c\in\Rpp$. We can separate the locations of $\phi^J(\x)$ and $\hat\err$ because Assumption \ref{ass:G} effectively sets the latter. 

To infer $G$ and $\phi^L$, we consider the distributions of the inter-wave durations $\Delta T^1$ and $\Delta T^2$ on the sequential exit events $\{T^2<\infty\}\bigcap\{{\cal E}^1=\{i\}, {\cal E}^2=\{-i\}\}$; $i=A,B$; across covariate values $\vx$.
From \eqref{eq:LapX}, these distributions are characterized by the Laplace transforms $\Lap^2\left(\mathbf{\s}^2;\left(\{i\},\{-i\}\right)\mid \vx\right)=$
\begin{equation}
\label{eq:LapOneTwo}
\mathbb{E}\left[\exp\left(-\Lambda(\s^1)\Delta\barYw^1(\vx)-\Lambda(\s^2)\Delta\barYw^2(\vx)\right)\mathds{1}_{E_i\left(\vx\right)}(\bigerr)\right],
\end{equation}
\noindent where $\Delta\barYw^1(\vx)=\phi^J(\x_i)+\err_i$ and $\Delta\barYw^2(\vx)=\phi^L(\x_{-i})-\phi^J(\x_i)+\err_{-i}-\err_i$ for $\bigerr$ in 
\begin{equation}
\label{eq:Ei}
E_i(\vx)\equiv E^2\left(\left(\{i\},\{-i\}\right);\vx\right)=\{\bigerr: \phi^J(\x_i)+\err_i<\phi^L(\x_{-i})+\err_{-i}\};
\end{equation}

\noindent $i=A,B$. To enable identification of $G$, we assume that, for appropriate choices of $\x^J$ and $\x^L$ in $\sX$, its support is covered by the sets $E_A(\x^J,\x^L)$ and $E_B(\x^L,\x^J)$ of $\bigerr$ for which we observe the two sequential exit patterns. 
\begin{assumption} \label{ass:overlap}
	There exist $\x^J,\x^L\in\sX$ such that $\phi^{J}(\x^J) \le \phi^{L}(\x^L)$.
\end{assumption}
 
\begin{theorem} \label{th:identTwo}
	Under Assumptions~\ref{ass:cont}--\ref{ass:overlap}, $c^{-1}\Lambda$, $c\phi^{L}$, $c\phi^{J}$, and the distribution $G(\cdot/c)$ of $c\bigerr$ are identified up to a scale factor $c\in\Rpp$.
\end{theorem}
\begin{proof}
	Lemma \ref{lemma:identTwo} implies identification of $c^{-1}\Lambda$, and thus $\salt\in[\s_0/c, \infty)\mapsto\psi(c\salt)$. Substituting $\s^1 = \psi(c\salt^1)$ and $\s^2 = \psi(c\salt^2)$ into (\ref{eq:LapOneTwo}) identifies $(\salt^1,\salt^2)\in[\s_{0}/c,\infty)^2\mapsto$
	\[
	\Lap^2\left(\psi(c\salt^1), \psi(c\salt^2); \left(\{i\},\{-i\}\right)\mid \vx\right)
	= \mathbb{E}\left[\exp\left(-\salt^1c\Delta\barYw^1(\vx)-\salt^2c\Delta\barYw^2(\vx)\right)\mathds{1}_{E_i(\vx)}(\bigerr)\right].
	\]
	\noindent This analytically extends to $(\salt^1,\salt^2)\in\Rp^2$ so, by uniqueness of the Laplace transform, identifies the distribution of $\left(c\Delta\barYw^1(\vx),c\Delta\barYw^2(\vx)\right)$ on $\left\{\bigerr\in E_i(\vx)\right\}$.
	
	For $\zeta\in\Rp$, $\Pr\left(c\Delta\barYw^2(\vx)>\zeta,\bigerr\in E_i(\vx)\right)$ identifies 
	\begin{align*}
	P_{i}(\zeta\mid\vx) \equiv \Pr\left(\zeta/c - \phi^L(\x_{-i})+\phi^J(\x_i)<\err_{-i}-\err_i\right).
	\end{align*} 

	\noindent Assumption \ref{ass:overlap} ensures we can find $\x^J$ and $\x^L$ such that $P_A(0\mid \x^J,\x^L)+P_B(0\mid \x^L,\x^J)\geq 1$ and thus $\phi^{J}(\x^J) \le \phi^{L}(\x^L)$. By Assumption \ref{ass:G}, $P_{A}(\zeta\mid \x^J, \x^L)+P_{B}(\zeta\mid\x^L, \x^J)$ strictly decreases in $\zeta$ and equals $\Pr(\err_A<\err_B)+\Pr(\err_B<\err_A)=1$ at $\zeta = c\phi^L(x^L)-c\phi^J(x^J)$ (which is nonnegative by Assumption \ref{ass:overlap}). Lemma \ref{lemma:identTwo} implies that $c\phi^J$, and thus $c\phi^{J}(x^J)$, is identified, so this determines $c\phi^{L}(x^L)$.

	We can repeat this argument to identify $c\phi^L(x')$ for any $\x'$ such that $P_A(0\mid \x^J,\x')>P_A(0\mid \x^J,\x^L)$ and thus $\phi^{J}(\x^J) \le \phi^{L}(\x^L)<\phi^L(\x')$. For all other $\x'$, $\phi^{L}(\x^L)\geq\phi^{L}(\x')$ and $P_A(c\phi^{L}(\x^L) - c\phi^{L}(\x')\mid \x^J, \x^L)=P_A(0\mid \x^J, \x')$. Because $P_A(\zeta\mid \x^J, \x^L)$ strictly decreases in $\zeta$ and $c\phi^{L}(\x^L)$ is known, this identifies $c\phi^{L}(\x')$.

	Given $c\phi^{J}(\x^J)$ and $c\phi^{L}(\x^L)$, the distributions of $\left(c\Delta\barYw^1(\x^J,\x^L),c\Delta\barYw^2(\x^J,\x^L)\right)$ on $\left\{\bigerr\in E_A(\x^J,\x^L)\right\}$ and
	$\left(c\Delta\barYw^1(\x^L,\x^J),c\Delta\barYw^2(\x^L,\x^J)\right)$ on $\left\{\bigerr\in E_B(\x^L,\x^J)\right\}$ identify the distribution of $c\bigerr$ on $\left\{\bigerr\in E_A(\x^J,\x^L)\bigcup E_B(\x^L,\x^J)\right\}$. As we chose $x^J$ and $x^L$ so that $\Pr\left(\bigerr\in  E_A(\x^J,\x^L)\bigcup E_B(\x^L,\x^J)\right)=1$, this identifies $G(\cdot/c)$.
%
\end{proof}

We can avoid Assumption \ref{ass:overlap} if players are randomly paired. Then, $\err_A$ and $\err_B$ are independent and identically distributed, so that their joint distribution $G$ follows immediately from the distribution of $\hat{\err}$ identified by Lemma \ref{lemma:identTwo}, up to the scale factor $c$. We can then use that, for given $\zeta$ and $\vx$, $P_i(\zeta|\vx)$ is a known monotone function of $c\phi^L(\x_{-i})$ to identify $c\phi^L$. 

\subsection{Identification with More than Two Players} 
\label{ss:iden2}

In the general case with $\np \geq 2$ players ${\cal I}=\left\{1,\ldots,I\right\}$, let the state $S$ again count the number of remaining players. We can directly extend the analysis for the two-player case to prove identification of $(\Lambda, \phi^1,\ldots,\phi^\np,G)$ from $\Lap^w\left(\cdot;\mathbf{e}^w\mid\vx\right)$ for all ordered partial partitions $\mathbf{e}^w$ of ${\cal I}$ and $\vx\in\sX^\np$; $w=1,\ldots,\np$. 

First, we again use the distribution of the first exit times $T^1$ in a game with identical covariate values $\vx=\left(x,\ldots,x\right)$ to identify $\phi^\np$ and $\Lambda$. As in the two-player case (Assumption \ref{ass:variation}), this requires variation in the initial thresholds.
\begin{assumption} \label{ass:variationGen}
$\phi^\np(\x)$ attains at least two distinct values for $\x\in\sX$. 
\end{assumption}

\noindent Next, we recurse the arguments in the proof of Theorem \ref{th:identTwo} to identify $\phi^{\np-1},\ldots,\phi^1$, and $G$. This requires a natural extension of Assumption \ref{ass:overlap}.
\begin{assumption} \label{ass:overlapGen}
There exist $x^1,\ldots,x^\np\in\sX$ such that $\phi^{\np}(x^\np) \leq\cdots\leq \phi^{1}(x^1)$.
\end{assumption}

\noindent In Appendix \ref{app:identGen}, we prove
\begin{theorem}\label{th:identGen}
	Under Assumptions \ref{ass:cont}, \ref{ass:G}, \ref{ass:variationGen}, and \ref{ass:overlapGen}, $c^{-1}\Lambda$; $c\phi^1,\ldots,c\phi^\np$; and $G(\cdot/c)$ are identified up to a scale factor $c\in\Rpp$.
\end{theorem}	

As in the two-player case, we do not need Assumption \ref{ass:overlapGen} if players are randomly matched. Then, $G(\cdot/c)$ is again determined from the distribution of $c\hat\err^\np$ identified from data on the first exit wave and $c\phi^\np,\ldots,c\phi^1$ easily follow.

\subsection{Extensions and Alternative Models}
\label{ss:extensions}

We showed in Section \ref{ss:spec} that the additively separable threshold specification follows naturally from an econometric implementation of Section \ref{s:example}'s example game. We specified the thresholds to be nonnegative so that all agents who are observed to start the game are initially satisfied to be in it. To review these assumptions, and the assumed independence of $\vX$ and $\bigerr$,  we need a perspective on how agents end up playing the game. We can obtain this by embedding the game in a model of selection into it. In general, this could be a dynamic matching model, in which agents repeatedly match up, possibly through costly search, to play the game. This would not only model selection into the game, with exit behavior governed by nonnegative thresholds, but also endogenize the value of foregoing or exiting the game.\footnote{\citet[][Section 4.3]{are10:abbring} discussed this in the context of single-agent optimal stopping.} This is an interesting setting to which our results can be adapted.

Developing this is beyond the scope of the paper, but we can provide some insight by considering a simpler, static selection model in which agents in some population of interest are randomly matched and play the synchronization game once if they find this profitable. Suppose that optimal stopping thresholds are additively separable (of course, this does not change because of selection into the game), but with $\mathbb{R}$-valued $\phi^S_i(\X_i)$ and $\err_i$. If following a match of, say, two agents $A$ and $B$, the game only starts with two players if $\Delta\barYw^1(\vX)=\min\left\{\phi^J(\X_A)+\err_A,\phi^J(\X_B)+\err_B\right\}>0$, then $\vX$ and $\bigerr$ are dependent across observed two-player games, even if they are not in the population of matches. Clearly, the results from the previous sections do not directly apply to this setting, but they can be combined with results on truncated regression models to establish identification. In particular, once we have identified $\Lambda$, an analysis along the lines of Theorem \ref{th:identTwo}'s proofs gives the distribution of $\Delta\barYw^1(\vx)\mid \Delta\barYw^1(\vx)>0$ for all $\vx=(\x_A,\x_B)\in\sX^2$. Specifically, for $\vx=(\x,\x)$, this identifies the distribution of $\Delta\barYw^1(\vx)=\phi^J(\x)+\hat\err$ conditional on $\phi^J(\x)+\hat\err>0$. We can apply \emph{e.g.} \citet[][Lemma 1]{chen_non-parametric_2010} to establish identification of $\phi ^{J}$. Then, identification of $\phi^L$ and $G$ follows as in the previous sections.

This analysis requires that we somehow identify $\Lambda$. In a single-game setup like the ones we have considered so far, this would have to rely on variation of the thresholds with the observed covariates $\vX$, using their independence of $\bigerr$ in the population of matches and the specific selection rule. In a grouped-game setup, in which we have multiple (say, two) observations of games that share the same value of $\bigerr$, but are driven by independent latent processes $\{Y\}$, we can identify $\Lambda$ without reference to (variation in) observed covariates. In particular, \citet[][Theorem 2]{abbring_mixed_2012} implies that $\Lambda$ is identified, even if we allow the latent processes and $\bigerr$ to generally depend on observed covariates, from the first exit times $T^1$. An important difference with \citeauthor{abbring_mixed_2012}'s single-agent case is that our games may generate various exit patterns. As these are fully determined by $\bigerr$ (and $\vX$), we have effectively restricted exit patterns to vary only across, and not within, groups. This may be easy to refute in data, in which case one needs to extend the model, \emph{e.g.} by allowing for some structured variation in $\bigerr$ within groups.  

In the single-game setup with observed covariates, we may want to consider alternatives to additively separable thresholds. Specifically, a multiplicatively separable specification $\overline{Y}_{i}^{S}=\phi^{S}\left(\X_{i}\right) \err_{i}$, with both factors nonnegative and independent, is natural if we want to enforce nonnegative threshold. \citet[][Section 3]{abbring_mixed_2012} gives examples of single-agent decision problems that lead to this specification (where the nonnegativity of the thresholds arises from the selection into the decision problem, as discussed above). \citet[][Theorem 1]{abbring_mixed_2012} can directly be used to prove a variant of Lemma \ref{lemma:identTwo} for this alternative specification, identifying $\phi^J$, $\Lambda$, and the distribution of $\hat\err$ up to scale from data on the first exit time.

\section{Estimation}
\label{s:estimation}

Suppose we have a random sample $\left(\mathbf{T}_n,\vX_n\right)$; $n=1,\ldots,N$; from the distribution of $(\mathbf{T},\vX)$ generated by Section \ref{ss:spec}'s game with additively separable thresholds. Specifically, we take $\{Y\}$ to be a Brownian motion ($\Pi=0$) with upward drift ($\mu\geq 0$), so that its hitting times are inverse Gaussian with Laplace exponent $\Lambda(\s)=\left(-\mu+\sqrt{\mu^2+2\s\sigma^2}\right)/\sigma^2$. We specify the thresholds to be linear in the $(1\times K)$-vector of covariates $\vX$ and the share $(\np-S)/\np$ of players who have exited: $\phi^{S}(\X_i)=\beta_0+\X_i\beta - \frac{\np-S}{\np}\delta$. Here, $\delta\geq 0$ captures the degree of strategic complementarity. The unobservable threshold terms $\bigerr$ follow an exchangeable multivariate lognormal distribution, with parameters $\mu^e$ (mean of $\ln \err_i$), $\sigma^e$ (standard deviation of $\ln \err_i$), and $\rho$ (correlation of $\ln\err_i$ and $\ln\err_j$ for $i\neq j$). We set $\sigma = 1$ to impose the scale normalization required for identification. The parameter vector to be estimated is $\alpha \equiv \left(\beta_0, \beta, \delta, \mu, \rho, \mu^{e}, \sigma^{e}\right)$.

We will present maximum simulated likelihood (MSL) and method of simulated moments (MSM) estimators of $\alpha$ and demonstrate their performance in simulation studies. The estimators can be adapted to alternative parametric specifications. This may have computational implications, so we will discuss this once we have the simulation evidence in place.

\subsection{Maximum Simulated Likelihood}
\label{ss:msle}

Consider the likelihood contribution of a game with outcome $\mathbf{T}$ given that its observed characteristics $\vX$ equal $\vx$ (we suppress the game subscript $n$ throughout this section). Given $\bigerr$, the game's exit pattern, including the number of waves $W$, is fully determined. Moreover, because $\{Y\}$ is a Brownian with upward drift, it hits the effective thresholds $\barYw^1(\vx), \ldots,\barYw^W(\vx)$ at finite times with probability one, so we observe all $W$ exit waves. Thus, we can represent the observed outcome $\mathbf{T}$ by $W$, the identities $\mathbf{e}^W$ of the players exiting in each wave, and the intra-wave durations $\mathbf{\Delta T}\equiv\left(\Delta T^1,\ldots\Delta T^W\right)$. By Section \ref{ss:spec}'s analysis, given $\bigerr\in E^W(\mathbf{e}^W;\vx)$ consistent with the exit pattern $\mathbf{e}^W$, these durations are independent and inverse Gaussian distributed with parameters determined by the corresponding inter-wave threshold differences $\Delta\barYw^1(\vx),\ldots, \Delta\barYw^W(\vx)$. So, the likelihood contribution of $\left(\mathbf{\Delta T},\mathbf{e}^W\right)$ given $\vx$ at parameter value $\alpha$ is
\begin{equation}\label{eq:mle}
	L(\alpha ; \mathbf{\Delta T},\mathbf{e}^W,\vx) = \int\prod_{w=1}^{W} f_{IG}\left( \Delta T^w \mid \Delta \barYw^w \left( \vx\right) \right) \mathds{1}_{E^W(\mathbf{e}^W;\vx)}(\bigerr)g(\bigerr)d\bigerr;
\end{equation}
where  
$
f_{IG}(t \mid y) = \frac{y}{\sqrt{2\pi t^3}} \exp\left( -\frac{\left( y - \mu t\right) ^2}{2 t} \right)
$
is the inverse Gaussian density with parameters $y/\mu$ and $y^2$. Note that the likelihood implicitly depends on $\mu^e$, $\sigma^e$, and $\rho$ through $g$; and on $\beta$, $\beta_0$, and $\delta$ through $E^W(\mathbf{e}^W;\vx)$ and $\Delta\barYw^w(\vx)$.


For a two-player game, constructing $L(\alpha ; \mathbf{\Delta T},\mathbf{e}^W,\vx)$
is easy. If the two firms exit sequentially, then $\mathbf{e}^W=\left(\{i\},\{-i\}\right)$, $E^W(\mathbf{e}^W;\vx)=E_i(\vx)$ defined in \eqref{eq:Ei},  and
\[
\begin{split}
&L(\alpha ; \mathbf{\Delta T},\mathbf{e}^W,\vx)=\\	 
&\int
	f_{IG}\left( \Delta T^1 \mid \phi^J(\x_i)+\err_i\right) 
	f_{IG}\left( \Delta T^2 \mid \phi^L(\x_{-i})-\phi^J(\x_i)+\err_{-i}-\err_i\right) 
	\mathds{1}_{E_i(\vx)}(\bigerr)g(\bigerr)d\bigerr;
\end{split}
\]

\noindent $i=A,B$. If they exit simultaneously, then $\mathbf{e}^W=\left(\{A,B\}\right)$ and $E^W(\mathbf{e}^W;\vx)=\overline{E}_A(\vx)\bigcup\overline{E}_B(\vx)$, where $\overline{E}_i(\vx)\equiv\left\{\bigerr\in\Rp^2:\phi^L(\x_{-i})\leq\phi^J(\x_i)+\err_i-\err_{-i}\leq\phi^J(\x_{-i})\right\}$ is the set of $\bigerr$ on which player $i\in\{A,B\}$ triggers joint exit, and
\[
L(\alpha ; \mathbf{\Delta T},\mathbf{e}^W,\vx) = \sum_{i=A,B}\int f_{IG}\left( \Delta T^1 \mid \phi^J(x_i)+\err_i\right)  \mathds{1}_{\overline{E}_i(\vx)}(\bigerr)g(\bigerr)d\bigerr.
\]

\noindent Importantly, we only need to consider one case--- set of $\bigerr$--- if exit is sequential (either $E_A(\vx)$ or $E_B(\vx)$) and two cases if it is simultaneous ($\overline{E}_A(\vx)$ and $\overline{E}_B(\vx)$). 

However, the number of such cases increases factorially in the number of players $\np$. In particular, $k\leq\np$ players exit simultaneously in an exit wave if they exit in any of their $k!$ possible orders, with the exit of each player (or set of players) triggering, through complementarities, the exit of the next one (or ones) in the same wave (see the discussion at the end of Section \ref{ss:mpg}). So, in general, evaluating $L(\alpha ; \mathbf{\Delta T},\mathbf{e}^W,\vx)$ requires evaluating $\np$-dimensional integrals over factorially many, often small, sets of $\bigerr$. Clearly, naive Monte Carlo integration would be highly inefficient. Instead, we will develop a Geweke-Hajivassiliou-Keane (GHK) style recursive simulator.\footnote{See \citet[][Section 3.2.2]{hajivassiliou1994classical} for details on the GHK simulator.} 

\subsubsection{Likelihood simulation}
\label{sss:likelihoodSim}

We will focus here on the fully exchangeable case in which $\vx=(\x,\ldots,\x)$, so that there is observed variation across games but not across players within games. This is a common empirical setting \citep[\emph{e.g.}][]{de2009inference} and therefore of interest in its own right. It suffices to make our point that MSL estimation works well in small enough problems, but quickly runs into computational limits if the number of players or observations increases.

In this fully exchangeable case, the likelihood contribution $L(\alpha ; \mathbf{\Delta T},\mathbf{e}^W,\vx)$ only depends on $\mathbf{e}^W$ through the number of players exiting in each wave--- the $W$-vector of \emph{wave sizes} $|\mathbf{e}^W|$--- not their identities. So, it is proportional to the likelihood $L(\alpha ; \mathbf{\Delta T},|\mathbf{e}^W|,\vx)$ of observing any exit pattern with wave sizes $|\mathbf{e}^W|$. Our GHK-style simulator computes $L(\alpha ; \mathbf{\Delta T},|\mathbf{e}^W|,\vx)$.

It uses that $\bigerr$ can be represented in terms of a standard normal common factor $\eta_0$ and independent standard normal idiosyncratic factors $\eta_1,\ldots,\eta_\np$:
\[
\frac{\ln\err_i-\mu_e}{\sigma_e} = \sqrt{\rho}\eta_0 + \sqrt{1-\rho}\eta_i,\quad i\in{\cal I}.
\]

\noindent Note that, for any given $S$, the ordering of the idiosyncratic factors $\eta_1,\ldots,\eta_\np$ aligns with the ordering of the thresholds $\phi^S(\x)+\err_1,\ldots,\phi^S(\x)+\err_\np$. 

We compute $L(\alpha ; \mathbf{\Delta T},|\mathbf{e}^W|,\vx)$ as an average over $\nsims$ simulation draws. For each $q = 1, \dots, \nsims$; we independently draw a common factor $\eta_{0,q}$ from the standard normal distribution and $\nu_{1,q}, \dots, \nu_{\np,q}$ from standard uniform distributions. From these, we recursively generate ordered idiosyncratic factors $\eta_q^{(1)} < \eta_q^{(2)} < \dots < \eta_q^{(\np)}$ and conditional exit pattern probabilities $P_q^{(1)},\ldots,P_q^{(\np)}$; recording wave-specific thresholds $\Delta\barYw_q^1,\ldots,\Delta\barYw_q^W$ along the way. 

First, we initialize this recursion by setting $P_q^{(1)}\equiv 1$ (as we always observe a first exit) and $\eta_q^{(1)} \equiv \Phi^{-1}\left( 1 - (1 - \nu_{1,q})^{1/\np} \right)$, with $\Phi$ the standard normal cumulative distribution function, equal to a draw from the minimum of $\np$ independent standard normal variables. We compute the corresponding unobserved threshold component 
$\hat\err_q^{1} \equiv \exp\left( \mu_e + \sigma_e\left( \sqrt{\rho}\eta_{0,q} + \sqrt{1-\rho}\eta_q^{(1)}\right)\right)$ and threshold $\Delta\barYw_q^1 \equiv \beta_0 + \x\beta + \hat\err_q^{1}$. 

Next, let $w_k$ be the wave in which the player with threshold rank $k$ exits and $l_k$ this player's rank within the wave. For example, if $|\mathbf{e}^W|=(2,1)$, then $w_1=w_2=1$, $w_3=2$, $l_1=1$, $l_2=2$, and $l_3=1$. Now, for $k=2,\ldots,\np$; repeat the following.
\begin{itemize}
\item  Given $\eta_q^{(1)},\ldots, \eta_q^{(k-1)}$, the next idiosyncratic factor $\eta_q^{(k)}$ is distributed as the minimum of $\np - k+1$ independent standard normal variables, left-truncated at $\eta_q^{(k-1)}$. Define the corresponding cumulative distribution function  
\[
	F_{k}(\eta \mid \eta^{(k-1)}) \equiv \frac{\Phi(-\eta^{(k-1)})^{{\np}-k+1} - \Phi(-\eta)^{{\np}-k+1}}{\Phi(-\eta^{(k-1)})^{{\np}-k+1}}.
\]

and the value of $\eta_q^{(k)}$ at which the $k$-th player would be indifferent between exiting in the ongoing wave $w_{k-1}$ and staying,
\[
	\overline{\eta}_q^{(k)} \equiv \frac{\ln\left(\hat\err^{w_{k-1}} + \delta l_{k-1}/\np\right) - \mu_e}{\sigma_e \sqrt{1-\rho}} - \sqrt{\frac{\rho}{1-\rho}}\eta_{0,q}.
\]

\item If $l_k>1$, the $k$-th player exits in the ongoing wave $w_k=w_{k-1}$. This requires that $\eta_q^{(k)}\leq\overline{\eta}_q^{(k)}$, which happens with conditional probability $P_q^{(k)} \equiv F_{k}(\overline{\eta}_q^{(k)} \mid \eta_q^{(k-1)})$. We generate $\eta_q^{(k)} \equiv F_{k}^{-1}\left(\nu_{k,q} P_q^{(k)} \mid \eta_q^{(k-1)}\right)$ to be consistent with this.

\item If $l_k=1$, the $k$-th player triggers a new wave $w_k$. This requires $\eta_q^{(k)} > \overline{\eta}_q^{(k)}$, which happens with conditional probability $P_q^{(k)} \equiv 1-F_{k}(\overline{\eta}_q^{(k)} \mid \eta_q^{(k-1)})$. We generate $\eta_q^{(k)} \equiv F_{k}^{-1}\left(\nu_{k,q} P_q^{(k)} + 1-P_q^{(k)} \mid \eta_q^{(k-1)}\right)$ to be consistent with this. We set $\hat\err_q^{w_k} \equiv \exp\left( \mu_e + \sigma_e\left( \sqrt{\rho}\eta_{0,q} + \sqrt{1-\rho}\eta_q^{(k)}\right)\right)$ and  $\Delta\barYw_q^{w_k} \equiv \hat\err_q^{w_k}-\hat\err_q^{w_k-1}$.
\end{itemize}

We combine the output from this procedure to generate likelihood contribution draws $L_q(\alpha ; \mathbf{\Delta T},|\mathbf{e}^W|,\vx) \equiv \left(\prod_{k=1}^{\np} P_q^{(k)}\right) \prod_{w=1}^{W} f_{IG}\left( \Delta T^w \mid \Delta \barYw_q^w\right)$; $q=1,\ldots,Q$. The procedure ensures these are smooth in $\alpha$ and strictly positive. Their average across $q$ gives the simulated likelihood contribution $\frac{1}{\nsims} \sum_{q=1}^\nsims L_q(\alpha ; \mathbf{\Delta T},|\mathbf{e}^W|,\vx)$. We repeat this, with independent simulation draws, for all $N$ observations, and sum the log contributions to construct the full simulated loglikelihood.  

\subsubsection{Performance of the MSL estimator}

\def\MSLEspOnebetaZeroMean {2.0461}
\def\MSLEspOnebetaMean {0.9956}
\def\MSLEspOnemuMean {1.0030}
\def\MSLEspOnedeltaMean {0.9531}
\def\MSLEspOnemuEMean {-0.0980}
\def\MSLEspOnesigmaEMean {1.0521}
\def\MSLEspOnerhoMean {0.1877}
\def\MSLEspOnebetaZeroStD {0.2060}
\def\MSLEspOnebetaStD {0.0689}
\def\MSLEspOnemuStD {0.0531}
\def\MSLEspOnedeltaStD {0.2174}
\def\MSLEspOnemuEStD {0.3356}
\def\MSLEspOnesigmaEStD {0.1733}
\def\MSLEspOnerhoStD {0.1683}
\def\MSLEspOnebetaZeroAStE {0.1921}
\def\MSLEspOnebetaAStE {0.0746}
\def\MSLEspOnemuAStE {0.0477}
\def\MSLEspOnedeltaAStE {0.1635}
\def\MSLEspOnemuEAStE {0.2583}
\def\MSLEspOnesigmaEAStE {0.1215}
\def\MSLEspOnerhoAStE {0.0954}
\def\MSLEspTwobetaZeroMean {1.9957}
\def\MSLEspTwobetaMean {1.0007}
\def\MSLEspTwomuMean {1.0100}
\def\MSLEspTwodeltaMean {1.0197}
\def\MSLEspTwomuEMean {0.0114}
\def\MSLEspTwosigmaEMean {0.9956}
\def\MSLEspTworhoMean {0.1816}
\def\MSLEspTwobetaZeroStD {0.2044}
\def\MSLEspTwobetaStD {0.0683}
\def\MSLEspTwomuStD {0.0519}
\def\MSLEspTwodeltaStD {0.1776}
\def\MSLEspTwomuEStD {0.2895}
\def\MSLEspTwosigmaEStD {0.1238}
\def\MSLEspTworhoStD {0.1398}
\def\MSLEspTwobetaZeroAStE {0.2102}
\def\MSLEspTwobetaAStE {0.0765}
\def\MSLEspTwomuAStE {0.0519}
\def\MSLEspTwodeltaAStE {0.1785}
\def\MSLEspTwomuEAStE {0.2768}
\def\MSLEspTwosigmaEAStE {0.1241}
\def\MSLEspTworhoAStE {0.1352}
\def\MSLEspThreebetaZeroMean {2.0157}
\def\MSLEspThreebetaMean {1.0088}
\def\MSLEspThreemuMean {1.0043}
\def\MSLEspThreedeltaMean {0.9897}
\def\MSLEspThreemuEMean {-0.0294}
\def\MSLEspThreesigmaEMean {1.0174}
\def\MSLEspThreerhoMean {0.1905}
\def\MSLEspThreebetaZeroStD {0.1017}
\def\MSLEspThreebetaStD {0.0347}
\def\MSLEspThreemuStD {0.0260}
\def\MSLEspThreedeltaStD {0.0881}
\def\MSLEspThreemuEStD {0.1527}
\def\MSLEspThreesigmaEStD {0.0631}
\def\MSLEspThreerhoStD {0.0951}
\def\MSLEspThreebetaZeroAStE {0.0962}
\def\MSLEspThreebetaAStE {0.0377}
\def\MSLEspThreemuAStE {0.0252}
\def\MSLEspThreedeltaAStE {0.0850}
\def\MSLEspThreemuEAStE {0.1340}
\def\MSLEspThreesigmaEAStE {0.0583}
\def\MSLEspThreerhoAStE {0.0735}
\def\MSLEspFourbetaZeroMean {2.0138}
\def\MSLEspFourbetaMean {0.9953}
\def\MSLEspFourmuMean {0.9980}
\def\MSLEspFourdeltaMean {0.9887}
\def\MSLEspFourmuEMean {-0.0172}
\def\MSLEspFoursigmaEMean {1.0085}
\def\MSLEspFourrhoMean {0.1956}
\def\MSLEspFourbetaZeroStD {0.1178}
\def\MSLEspFourbetaStD {0.0613}
\def\MSLEspFourmuStD {0.0351}
\def\MSLEspFourdeltaStD {0.0693}
\def\MSLEspFourmuEStD {0.0902}
\def\MSLEspFoursigmaEStD {0.0476}
\def\MSLEspFourrhoStD {0.0465}
\def\MSLEspFourbetaZeroAStE {0.1108}
\def\MSLEspFourbetaAStE {0.0672}
\def\MSLEspFourmuAStE {0.0348}
\def\MSLEspFourdeltaAStE {0.0691}
\def\MSLEspFourmuEAStE {0.0818}
\def\MSLEspFoursigmaEAStE {0.0426}
\def\MSLEspFourrhoAStE {0.0389}
\def\MSLEspOneI {2}
\def\MSLEspTwoI {2}
\def\MSLEspThreeI {2}
\def\MSLEspFourI {5}
\def\MSLEspOneN {500}
\def\MSLEspTwoN {500}
\def\MSLEspThreeN {2{,}000}
\def\MSLEspFourN {500}
\def\MSLEspOneR {1{,}000}
\def\MSLEspTwoR {10{,}000}
\def\MSLEspThreeR {10{,}000}
\def\MSLEspFourR {10{,}000}
\def\MSLEspOnetime {2{,}185.6}
\def\MSLEspTwotime {23{,}857.1}
\def\MSLEspThreetime {91{,}434.2}
\def\MSLEspFourtime {67{,}419.9}
\def\MSLEspOneflag {0}
\def\MSLEspTwoflag {0}
\def\MSLEspThreeflag {0}
\def\MSLEspFourflag {0}
\def\MSLEsptruebetaZero {2.0}
\def\MSLEsptruebeta {1.0}
\def\MSLEsptruemu {1.0}
\def\MSLEsptruedelta {1.0}
\def\MSLEsptruemuE {0.0}
\def\MSLEsptruesigmaE {1.0}
\def\MSLEsptruerho {0.2}

\def\MSMspOnebetaZeroMean {1.9376}
\def\MSMspOnebetaMean {1.0218}
\def\MSMspOnemuMean {1.0316}
\def\MSMspOnedeltaMean {1.0689}
\def\MSMspOnemuEMean {0.1178}
\def\MSMspOnesigmaEMean {0.9662}
\def\MSMspOnerhoMean {0.2251}
\def\MSMspOnebetaZeroStD {0.2249}
\def\MSMspOnebetaStD {0.0978}
\def\MSMspOnemuStD {0.0597}
\def\MSMspOnedeltaStD {0.1923}
\def\MSMspOnemuEStD {0.2827}
\def\MSMspOnesigmaEStD {0.1236}
\def\MSMspOnerhoStD {0.1664}
\def\MSMspOnebetaZeroAStE {0.2478}
\def\MSMspOnebetaAStE {0.1051}
\def\MSMspOnemuAStE {0.0688}
\def\MSMspOnedeltaAStE {0.2045}
\def\MSMspOnemuEAStE {0.3233}
\def\MSMspOnesigmaEAStE {0.1298}
\def\MSMspOnerhoAStE {0.1827}
\def\MSMspTwobetaZeroMean {2.0015}
\def\MSMspTwobetaMean {1.0061}
\def\MSMspTwomuMean {1.0106}
\def\MSMspTwodeltaMean {1.0122}
\def\MSMspTwomuEMean {-0.0043}
\def\MSMspTwosigmaEMean {0.9921}
\def\MSMspTworhoMean {0.1578}
\def\MSMspTwobetaZeroStD {0.1744}
\def\MSMspTwobetaStD {0.0899}
\def\MSMspTwomuStD {0.0458}
\def\MSMspTwodeltaStD {0.2171}
\def\MSMspTwomuEStD {0.2801}
\def\MSMspTwosigmaEStD {0.2042}
\def\MSMspTworhoStD {0.1138}
\def\MSMspTwobetaZeroAStE {0.1773}
\def\MSMspTwobetaAStE {0.0902}
\def\MSMspTwomuAStE {0.0469}
\def\MSMspTwodeltaAStE {0.2406}
\def\MSMspTwomuEAStE {0.2817}
\def\MSMspTwosigmaEAStE {0.2146}
\def\MSMspTworhoAStE {0.1430}
\def\MSMspThreebetaZeroMean {2.0068}
\def\MSMspThreebetaMean {0.9997}
\def\MSMspThreemuMean {0.9987}
\def\MSMspThreedeltaMean {0.9955}
\def\MSMspThreemuEMean {-0.0137}
\def\MSMspThreesigmaEMean {1.0055}
\def\MSMspThreerhoMean {0.1937}
\def\MSMspThreebetaZeroStD {0.0312}
\def\MSMspThreebetaStD {0.0096}
\def\MSMspThreemuStD {0.0077}
\def\MSMspThreedeltaStD {0.0294}
\def\MSMspThreemuEStD {0.0537}
\def\MSMspThreesigmaEStD {0.0200}
\def\MSMspThreerhoStD {0.0281}
\def\MSMspThreebetaZeroAStE {0.0328}
\def\MSMspThreebetaAStE {0.0117}
\def\MSMspThreemuAStE {0.0084}
\def\MSMspThreedeltaAStE {0.0301}
\def\MSMspThreemuEAStE {0.0563}
\def\MSMspThreesigmaEAStE {0.0215}
\def\MSMspThreerhoAStE {0.0287}
\def\MSMspFourbetaZeroMean {1.9992}
\def\MSMspFourbetaMean {1.0013}
\def\MSMspFourmuMean {1.0003}
\def\MSMspFourdeltaMean {0.9998}
\def\MSMspFourmuEMean {-0.0018}
\def\MSMspFoursigmaEMean {0.9996}
\def\MSMspFourrhoMean {0.1985}
\def\MSMspFourbetaZeroStD {0.0104}
\def\MSMspFourbetaStD {0.0090}
\def\MSMspFourmuStD {0.0045}
\def\MSMspFourdeltaStD {0.0196}
\def\MSMspFourmuEStD {0.0174}
\def\MSMspFoursigmaEStD {0.0085}
\def\MSMspFourrhoStD {0.0071}
\def\MSMspFourbetaZeroAStE {0.0109}
\def\MSMspFourbetaAStE {0.0090}
\def\MSMspFourmuAStE {0.0046}
\def\MSMspFourdeltaAStE {0.0185}
\def\MSMspFourmuEAStE {0.0162}
\def\MSMspFoursigmaEAStE {0.0089}
\def\MSMspFourrhoAStE {0.0074}
\def\MSMspRhoDrop {32}
\def\MSMspTimeInflation {3}
\def\MSMspOneI {2}
\def\MSMspTwoI {5}
\def\MSMspThreeI {2}
\def\MSMspFourI {20}
\def\MSMspOneN {500}
\def\MSMspTwoN {500}
\def\MSMspThreeN {50{,}000}
\def\MSMspFourN {50{,}000}
\def\MSMspOneR {100}
\def\MSMspTwoR {100}
\def\MSMspThreeR {1}
\def\MSMspFourR {1}
\def\MSMspOnetime {402.3}
\def\MSMspTwotime {470.3}
\def\MSMspThreetime {545.7}
\def\MSMspFourtime {1{,}277.4}
\def\MSMspOneflag {0}
\def\MSMspTwoflag {0}
\def\MSMspThreeflag {0}
\def\MSMspFourflag {0}
\def\MSMsptruebetaZero {2.0}
\def\MSMsptruebeta {1.0}
\def\MSMsptruemu {1.0}
\def\MSMsptruedelta {1.0}
\def\MSMsptruemuE {0.0}
\def\MSMsptruesigmaE {1.0}
\def\MSMsptruerho {0.2}

Our MSL estimator maximizes the simulated loglikelihood using, to avoid convergence to local maxima, a multi-start routine with five random starting values. Under standard regularity conditions, if $\nsims/\sqrt{N}\rightarrow\infty$ as $N\rightarrow\infty$, it is asymptotically equivalent to the maximum likelihood estimator \citep[][Proposition 3.2]{gourieroux1996simulation} and the outer-product-of-the-gradient estimator of its asymptotic variance-covariance matrix is consistent. We use this estimator to calculate asymptotic standard errors.

\begin{table}[tbp]
	\caption{Monte Carlo Experiments for MSL}
	\label{table:MSLE}
	\begin{center}
		\vspace*{3mm}
		\begin{tabular}{@{} l c c ccc c ccc @{}}
			\toprule
			&                      & & \multicolumn{3}{c}{\MCone}   & & \multicolumn{3}{c}{\MCtwo}   \\ 
			\cmidrule(lr){4-6} \cmidrule(l){8-10}  
			&                      & & Mean & StD & AStE & & Mean & StD & AStE \\ 
			\cmidrule(lr){4-6} \cmidrule(l){8-10} 
			$\beta_0$  & \MSLEsptruebetaZero   & & \MSLEspOnebetaZeroMean & \MSLEspOnebetaZeroStD & \MSLEspOnebetaZeroAStE & & \MSLEspTwobetaZeroMean & \MSLEspTwobetaZeroStD & \MSLEspTwobetaZeroAStE \\
			$\beta$    & \MSLEsptruebeta       & & \MSLEspOnebetaMean     & \MSLEspOnebetaStD     & \MSLEspOnebetaAStE     & & \MSLEspTwobetaMean     & \MSLEspTwobetaStD     & \MSLEspTwobetaAStE     \\
			$\mu$      & \MSLEsptruemu         & & \MSLEspOnemuMean       & \MSLEspOnemuStD       & \MSLEspOnemuAStE       & & \MSLEspTwomuMean       & \MSLEspTwomuStD       & \MSLEspTwomuAStE       \\
			$\delta$   & \MSLEsptruedelta      & & \MSLEspOnedeltaMean    & \MSLEspOnedeltaStD    & \MSLEspOnedeltaAStE    & & \MSLEspTwodeltaMean    & \MSLEspTwodeltaStD    & \MSLEspTwodeltaAStE    \\ 
			$\mu^e$    & \MSLEsptruemuE        & & \MSLEspOnemuEMean      & \MSLEspOnemuEStD      & \MSLEspOnemuEAStE      & & \MSLEspTwomuEMean      & \MSLEspTwomuEStD      & \MSLEspTwomuEAStE      \\
			$\sigma^e$ & \MSLEsptruesigmaE     & & \MSLEspOnesigmaEMean   & \MSLEspOnesigmaEStD   & \MSLEspOnesigmaEAStE   & & \MSLEspTwosigmaEMean   & \MSLEspTwosigmaEStD   & \MSLEspTwosigmaEAStE   \\
			$\rho$     & \MSLEsptruerho        & & \MSLEspOnerhoMean      & \MSLEspOnerhoStD      & \MSLEspOnerhoAStE      & & \MSLEspTworhoMean      & \MSLEspTworhoStD      & \MSLEspTworhoAStE      \\ 
			\cmidrule(lr){4-6} \cmidrule(l){8-10} 
			$\np$      & \multicolumn{1}{c}{} & \multicolumn{1}{c}{} & \multicolumn{3}{c}{\MSLEspOneI} & & \multicolumn{3}{c}{\MSLEspTwoI} \\
			$N$        & \multicolumn{1}{c}{} & \multicolumn{1}{c}{} & \multicolumn{3}{c}{\MSLEspOneN} & & \multicolumn{3}{c}{\MSLEspTwoN} \\
			$\nsims$   & \multicolumn{1}{c}{} & \multicolumn{1}{c}{} & \multicolumn{3}{c}{\MSLEspOneR} & & \multicolumn{3}{c}{\MSLEspTwoR} \\  
			time (sec.)& \multicolumn{1}{c}{} & \multicolumn{1}{c}{} & \multicolumn{3}{c}{\MSLEspOnetime}& & \multicolumn{3}{c}{\MSLEspTwotime} \\ 
			
			\\[2ex]
			
			&                      & & \multicolumn{3}{c}{\MCthree} & & \multicolumn{3}{c}{\MCfour} \\ 
			\cmidrule(lr){4-6} \cmidrule(l){8-10}  
			&                      & & Mean & StD & AStE & & Mean & StD & AStE \\ 
			\cmidrule(lr){4-6} \cmidrule(l){8-10} 
			$\beta_0$  & \MSLEsptruebetaZero   & & \MSLEspThreebetaZeroMean & \MSLEspThreebetaZeroStD & \MSLEspThreebetaZeroAStE & & \MSLEspFourbetaZeroMean & \MSLEspFourbetaZeroStD & \MSLEspFourbetaZeroAStE \\
			$\beta$    & \MSLEsptruebeta       & & \MSLEspThreebetaMean     & \MSLEspThreebetaStD     & \MSLEspThreebetaAStE     & & \MSLEspFourbetaMean     & \MSLEspFourbetaStD     & \MSLEspFourbetaAStE     \\
			$\mu$      & \MSLEsptruemu         & & \MSLEspThreemuMean       & \MSLEspThreemuStD       & \MSLEspThreemuAStE       & & \MSLEspFourmuMean       & \MSLEspFourmuStD       & \MSLEspFourmuAStE       \\
			$\delta$   & \MSLEsptruedelta      & & \MSLEspThreedeltaMean    & \MSLEspThreedeltaStD    & \MSLEspThreedeltaAStE    & & \MSLEspFourdeltaMean    & \MSLEspFourdeltaStD    & \MSLEspFourdeltaAStE    \\ 
			$\mu^e$    & \MSLEsptruemuE        & & \MSLEspThreemuEMean      & \MSLEspThreemuEStD      & \MSLEspThreemuEAStE      & & \MSLEspFourmuEMean      & \MSLEspFourmuEStD      & \MSLEspFourmuEAStE      \\
			$\sigma^e$ & \MSLEsptruesigmaE     & & \MSLEspThreesigmaEMean   & \MSLEspThreesigmaEStD   & \MSLEspThreesigmaEAStE   & & \MSLEspFoursigmaEMean   & \MSLEspFoursigmaEStD   & \MSLEspFoursigmaEAStE   \\
			$\rho$     & \MSLEsptruerho        & & \MSLEspThreerhoMean      & \MSLEspThreerhoStD      & \MSLEspThreerhoAStE      & & \MSLEspFourrhoMean      & \MSLEspFourrhoStD      & \MSLEspFourrhoAStE      \\ 
			\cmidrule(lr){4-6} \cmidrule(l){8-10} 
			$\np$      & \multicolumn{1}{c}{} & \multicolumn{1}{c}{} & \multicolumn{3}{c}{\MSLEspThreeI} & & \multicolumn{3}{c}{\MSLEspFourI} \\
			$N$        & \multicolumn{1}{c}{} & \multicolumn{1}{c}{} & \multicolumn{3}{c}{\MSLEspThreeN} & & \multicolumn{3}{c}{\MSLEspFourN} \\
			$\nsims$   & \multicolumn{1}{c}{} & \multicolumn{1}{c}{} & \multicolumn{3}{c}{\MSLEspThreeR} & & \multicolumn{3}{c}{\MSLEspFourR} \\  
			time (sec.)& \multicolumn{1}{c}{} & \multicolumn{1}{c}{} & \multicolumn{3}{c}{\MSLEspThreetime}& & \multicolumn{3}{c}{\MSLEspFourtime} \\ 
			\bottomrule
		\end{tabular}
	\end{center}
	\vspace*{3mm}
	\scriptsize{	
		Note: This table reports the means, standard deviations (StD), and average asymptotic standard errors (AStE) of MSL estimates across 100 Monte Carlo samples for four simulation designs, \MCone--\MCfour. The true values of the estimated parameters are fixed across designs and listed in the second column. The covariates $\vX=(\X,\ldots,\X)$, with $\X$ scalar ($K=1$) and standard lognormal. The designs vary in the number of players $\np$, sample size $N$, and number of simulation draws $\nsims$ used by the estimator. \computeplatform 
	}
\end{table}

Table \ref{table:MSLE} presents Monte Carlo evidence on the finite-sample performance of the MSL estimator. We consider four designs, \MCone--\MCfour, that vary in the number of players $\np$, sample size $N$, and number of simulation draws $\nsims$, but all use the same true parameter values $\left( \beta_0, \beta, \mu, \delta, \mu^{e}, \sigma^{e}, \rho\right)  = \left( \MSLEsptruebetaZero, \MSLEsptruebeta, \MSLEsptruemu, \MSLEsptruedelta, \MSLEsptruemuE, \MSLEsptruesigmaE, \MSLEsptruerho\right)$, where the covariates $\vX=(\X,\ldots,\X)$ with $\X$ scalar ($K=1$) and standard lognormal. For each design, we report the means, standard deviations (StD), and average asymptotic standard errors (AStE) of the estimates across 100 Monte Carlo samples. We also report the mean computation times per sample, which cover data generation, sequential multi-start optimization, and calculation of the standard errors. These are conservative, as we could speed the procedure up by parallelizing the multi-start routine and using more powerful hardware (see the note to the table).

Panel \MCone\ reports on a small simulation design with $\np=\MSLEspOneI$ players, $N=\MSLEspOneN$ observations, and $\nsims=\MSLEspOneR$ simulation draws. This uses limited computational resources, but estimates all parameters except $\beta$ and $\mu$ with nonnegligible bias. Moreover, the asymptotic standard errors for all but $\beta$, $\mu$, and $\beta_0$ severely underestimate the standard deviations of the estimates across samples. 

A comparison with Design \MCtwo, which uses 10 times as many simulation draws ($\nsims=\MSLEspTwoR$), shows that this is primarily due to simulation error. In particular, in Panel \MCtwo, the biases mostly disappear and the gaps between the asymptotic standard errors and standard deviations close. Computation times, however, increase about linearly with $\nsims$ to well over 6 hours per sample.

Design \MCthree\ increases the sample size fourfold to \(N=\MSLEspThreeN\) while keeping \(\nsims=\MSLEspThreeR\). As to be expected, the standard errors decrease to about $1/\sqrt{4} = 1/2$ of their values in Panel \MCtwo. To a limited extent, the biases and the discrepancies between the standard deviations and asymptotic standard errors reappear. If we take \citeauthor{gourieroux1996simulation}'s condition for asymptotic equivalence that $\nsims$ increases faster than $\sqrt{N}$ as a guide, we should at least double the number of simulation draws, which would more than double the 25-hour computation time.

Finally, Design \MCfour\ returns to the smaller samples of Design \MCtwo\ ($N=\MSLEspFourN$), but for games with $\np=\MSLEspFourI$ players. With more players, there is richer variation in exit patterns and times to identify the parameters, but the simulator needs to calculate higher dimensional integrals, using longer recursions. Overall, the estimator for this larger game is much more precise, but the increased demands on the simulator show up in some minor biases and a tripling of computation times to over 18 hours.  

From all this, we conclude that MSL can estimate the parameters of our model with small bias and reasonable precision, and that our asymptotic standard errors are a reliable guide to inference, in small enough problems.\footnote{In \citet[][Section S5.1]{ecsgsupp26:abbringyu}, we explore whether jackknife debiasing would allow us to use fewer simulation draws and thus handle larger problems, but conclude it does not.} This straightforwardly extends to some alternative parametric specifications, in particular of $\phi^S$. Other changes to the model bring nontrivial additional computation. For example, we can generalize the latent process $\{Y\}$ to any parametric L\'{e}vy processes with negative shocks, but this would require that we replace the inverse Gaussian expressions for the hitting-time densities with a version of \cites{jem21:abbringsalimans} procedure for calculating them. Similarly, if we relax  exchangeability of the thresholds, we will have to resort to a more complicated simulator.

\subsection{Method of Simulated Moments}\label{ss:msm}

MSL can be used in small enough problems and provides, under conditions that ensure asymptotic equivalence with maximum likelihood, an efficient benchmark. However, although there is scope for speeding up the MSL procedure and pushing the computational boundaries a bit beyond those suggested by Table \ref{table:MSLE}, MSL computation quickly becomes infeasible if the numbers of players and observations grow. Richer specification of the latent process and thresholds would further limit the scale of the problems that MSL can handle. Our MSM estimator is a computationally efficient alternative for these cases, because it is $\sqrt{N}$-consistent for any fixed number of simulation draws $\nsims$ per observation, including $\nsims=1$. 

\subsubsection{Moments}

We will again focus on the fully exchangeable case in which $\vX_n$ takes values $\vx_n=(\x_n,\ldots,\x_n)$; $n=1,\ldots,N$ (here, it is convenient to keep track of the game subscript $n$). In this case, we can focus on \emph{ordered} exit times $T_n^{(1)} \le T_n^{(2)} \le \dots \le T_n^{(\np)}$ without losing structural information. We consider the following moments.
\begin{enumerate}
     \item \textbf{Sequential exit:} We match $m_1(\vx_n,\mathbf{T}_n)\equiv\tilde{\x}_n\mathds{1}\left\{T_n^{(2)} > T_n^{(1)}\right\}$, where $\widetilde{\x}_n\equiv\left(1~\x_n\right)$ is a $1\times(1+K)$ vector, to its population equivalent $m^*_1(\vx_n;\alpha)$. Exchangeability implies that $\mathds{1}\left\{T_n^{(2)} > T_n^{(1)}\right\}$ only depends on $\bigerr_n$ and is independent of $\widetilde{\vX}_n$. So, $m^*_1(\vx_n;\alpha)=\tilde{\x}_n\Pr(T_n^{(2)} > T_n^{(1)})$ and, using Section \ref{sss:likelihoodSim}'s factorization of $\bigerr_n$, we can express $\Pr(T_n^{(2)} > T_n^{(1)})=$ 
    \[
			\int\int \np(\np-1) \phi(\eta_0)\phi(\eta^{(2)})[\Phi(-\eta^{(2)})]^{I-2} \Phi \left(\tfrac{\ln(\varepsilon^{(2)} - \delta/\np) - \mu_e - \sigma_e\sqrt{\rho}\eta_0}{\sigma_e\sqrt{1-\rho}}\right) d\eta^{(2)}d\eta_0,
    \]
    where $\varepsilon^{(2)}\equiv\exp(\mu_e + \sigma_e(\sqrt{\rho}\eta_0 + \sqrt{1-\rho}\eta^{(2)}))$. We evaluate this integral numerically to calculate $m^*_1(\vx_n;\alpha)$. This way, we avoid the discontinuity in $\alpha$ that would arise from simulating over the discrete indicator $m_1(\vx_n,\cdot)$.
	
	We calculate all remaining population moments by simulation.

	\item \textbf{Wave count:} We match $m_2(\vx_n,\mathbf{T}_n) \equiv \tilde{\x}_n\widetilde{W}_n$, where
	   \[
        	\widetilde{W}_n\equiv 1 + \sum_{k=2}^{I} \bigl(1 - \exp\left(-50\left(T^{(k)} - T^{(k-1)}\right)\right)\bigr)
    	\]
    	is a smoothed count of the number of waves, to a simulated version of it, $m^*_2(\vx_n;\alpha)\equiv\frac{1}{\nsims}\sum_{q=1}^{\nsims}m_2\left(\vx_n,\mathbf{T}_{n,q}(\alpha)\right)$. Here, $\mathbf{T}_{n,q}(\alpha)$; $q=1,\ldots,\nsims$; are simulated exit times for a game with covariates $\vx_n$ and parameters $\alpha$. 
		
		The probability of sequential exit and wave count are particularly informative about complementarities ($\delta$) and correlation of the unobservables ($\rho$). 

	\item \textbf{Hitting times:} 
		Like \citet{de2009inference}, we 	include mean and harmonic mean (or, equivalently, mean of the reciprocal) hitting times because these are sufficient statistics for the parameters of the inverse Gaussian distribution. So, we align $m_3(\vx_n,\mathbf{T}_n)\equiv\tilde{\x}_n T_n^{(1)}$, $m_4(\vx_n,\mathbf{T}_n)\equiv\tilde{\x}_n T_n^{(2)}$, and $m_5(\vx_n,\mathbf{T}_n)\equiv\tilde{\x}_n/T_n^{(1)}$ to their simulated counterparts $m_3^*(\vx_n;\alpha)$, $m_4^*(\vx_n;\alpha)$, and $m_5^*(\vx_n;\alpha)$. 
   
    \item \textbf{Laplace transforms:} 
		We capture aspects of the distribution of the first two exit times beyond their means with
		$m_6(\vx_n,\mathbf{T}_n)\equiv\tilde{\x}_n\exp(-T_n^{(1)})$, $m_7(\vx_n,\mathbf{T}_n)\equiv\tilde{\x}_n\exp(-T_n^{(2)})$, and $m_8(\vx_n,\mathbf{T}_n)\equiv\tilde{\x}_n\exp(-0.6 T_n^{(1)} - 0.4 T_n^{(2)})$; and the game's total duration $T_n^{(\np)}$ with $m_9(\vx_n,\mathbf{T}_n)\equiv\tilde{\x}_n\exp(-0.5 T_n^{(\np)})$. We match these to their simulated analogues $m_6^*(\vx_n;\alpha)$, $m_7^*(\vx_n;\alpha)$, $m_8^*(\vx_n;\alpha)$, and $m_9^*(\vx_n;\alpha)$.  
\end{enumerate}

\noindent We collect these $9(1+K)$ moments in vectors
$m(\vx_n,\mathbf{T}_n)\equiv(m_1(\vx_n,\mathbf{T}_n), \ldots, m_9(\vx_n,\mathbf{T}_n))'$
and 
$m^*(\vx_n;\alpha)\equiv(m_1^*(\vx_n;\alpha), \ldots, m_9^*(\vx_n;\alpha))'$.
Our MSM estimator minimizes 
\[\left[\frac{1}{N} \sum_{n=1}^{N} \left( m\left(\mathbf{x}_n, \mathbf{T}_n\right) - m^*\left(\mathbf{x}_n;\alpha\right) \right) \right]'\mathbb{W}\left[\frac{1}{N} \sum_{n=1}^{N} \left( m\left(\mathbf{x}_n, \mathbf{T}_n\right) - m^*\left(\mathbf{x}_n;\alpha\right) \right) \right],
\]
where the weighting matrix $\mathbb{W}$ has the inverse sample variances of $m(\vx_n,\mathbf{T}_n)$ on its diagonal and zeros elsewhere.\footnote{This is a one-step estimator. In \citet[][Section S5.2]{ecsgsupp26:abbringyu} we show that an optimal two-step version of this estimator, if anything, performs slightly worse in finite samples.} We use the corresponding sandwich variance-covariance estimator to compute asymptotic standard errors.

\subsubsection{Performance of the MSM estimator}

Under standard regularity conditions, the MSM estimator is $\sqrt{N}$-consistent and asymptotically normal for any fixed $\nsims$ \citep[][Proposition 2.3]{gourieroux1996simulation}. It is asymptotically equivalent to the generalized method of moments estimator that uses the population moments $\mathbb{E}\left[m(\vx_n,\mathbf{T}_n)\right]$ instead of their simulated counterparts $m^*(\vx_n;\alpha)$ if $\nsims\rightarrow\infty$ as $N\rightarrow\infty$. The sandwich estimator of the asymptotic variance-covariance matrix is consistent for fixed $\nsims$.

\begin{table}[tbp]
	\caption{Monte Carlo Experiments for One-Step MSM}
	\label{table:MSM}
	\begin{center}
		\vspace*{3mm}
		\begin{tabular}{@{} l c c ccc c ccc @{}}
			\toprule
			&                      & & \multicolumn{3}{c}{\MCone}   & & \multicolumn{3}{c}{\MCtwo}   \\ 
			\cmidrule(lr){4-6} \cmidrule(l){8-10}  
			&                      & & Mean & StD & AStE & & Mean & StD & AStE \\ 
			\cmidrule(lr){4-6} \cmidrule(l){8-10} 
			
			$\beta_0$  & \MSMsptruebetaZero   & & \MSMspOnebetaZeroMean & \MSMspOnebetaZeroStD & \MSMspOnebetaZeroAStE & & \MSMspTwobetaZeroMean & \MSMspTwobetaZeroStD & \MSMspTwobetaZeroAStE \\
			$\beta$    & \MSMsptruebeta       & & \MSMspOnebetaMean     & \MSMspOnebetaStD     & \MSMspOnebetaAStE     & & \MSMspTwobetaMean     & \MSMspTwobetaStD     & \MSMspTwobetaAStE     \\
			$\mu$      & \MSMsptruemu         & & \MSMspOnemuMean       & \MSMspOnemuStD       & \MSMspOnemuAStE       & & \MSMspTwomuMean       & \MSMspTwomuStD       & \MSMspTwomuAStE       \\
			$\delta$   & \MSMsptruedelta      & & \MSMspOnedeltaMean    & \MSMspOnedeltaStD    & \MSMspOnedeltaAStE    & & \MSMspTwodeltaMean    & \MSMspTwodeltaStD    & \MSMspTwodeltaAStE    \\ 
			$\mu^e$    & \MSMsptruemuE        & & \MSMspOnemuEMean      & \MSMspOnemuEStD      & \MSMspOnemuEAStE      & & \MSMspTwomuEMean      & \MSMspTwomuEStD      & \MSMspTwomuEAStE      \\
			$\sigma^e$ & \MSMsptruesigmaE     & & \MSMspOnesigmaEMean   & \MSMspOnesigmaEStD   & \MSMspOnesigmaEAStE   & & \MSMspTwosigmaEMean   & \MSMspTwosigmaEStD   & \MSMspTwosigmaEAStE   \\
			$\rho$     & \MSMsptruerho        & & \MSMspOnerhoMean      & \MSMspOnerhoStD      & \MSMspOnerhoAStE      & & \MSMspTworhoMean      & \MSMspTworhoStD      & \MSMspTworhoAStE      \\ 
			\cmidrule(lr){4-6} \cmidrule(l){8-10} 
			
			$\np$      & \multicolumn{1}{c}{} & \multicolumn{1}{c}{} & \multicolumn{3}{c}{\MSMspOneI} & & \multicolumn{3}{c}{\MSMspTwoI} \\
			$N$        & \multicolumn{1}{c}{} & \multicolumn{1}{c}{} & \multicolumn{3}{c}{\MSMspOneN} & & \multicolumn{3}{c}{\MSMspTwoN} \\
			$\nsims$   & \multicolumn{1}{c}{} & \multicolumn{1}{c}{} & \multicolumn{3}{c}{\MSMspOneR} & & \multicolumn{3}{c}{\MSMspTwoR} \\  
			time (sec.)& \multicolumn{1}{c}{} & \multicolumn{1}{c}{} & \multicolumn{3}{c}{\MSMspOnetime}& & \multicolumn{3}{c}{\MSMspTwotime} \\ 
			
			\\[2ex]
			
			&                      & & \multicolumn{3}{c}{\MCthree} & & \multicolumn{3}{c}{\MCfour} \\ 
			\cmidrule(lr){4-6} \cmidrule(l){8-10}  
			&                      & & Mean & StD & AStE & & Mean & StD & AStE \\ 
			\cmidrule(lr){4-6} \cmidrule(l){8-10} 
			
			$\beta_0$  & \MSMsptruebetaZero   & & \MSMspThreebetaZeroMean & \MSMspThreebetaZeroStD & \MSMspThreebetaZeroAStE & & \MSMspFourbetaZeroMean & \MSMspFourbetaZeroStD & \MSMspFourbetaZeroAStE \\
			$\beta$    & \MSMsptruebeta       & & \MSMspThreebetaMean     & \MSMspThreebetaStD     & \MSMspThreebetaAStE     & & \MSMspFourbetaMean     & \MSMspFourbetaStD     & \MSMspFourbetaAStE     \\
			$\mu$      & \MSMsptruemu         & & \MSMspThreemuMean       & \MSMspThreemuStD       & \MSMspThreemuAStE       & & \MSMspFourmuMean       & \MSMspFourmuStD       & \MSMspFourmuAStE       \\
			$\delta$   & \MSMsptruedelta      & & \MSMspThreedeltaMean    & \MSMspThreedeltaStD    & \MSMspThreedeltaAStE    & & \MSMspFourdeltaMean    & \MSMspFourdeltaStD    & \MSMspFourdeltaAStE    \\ 
			$\mu^e$    & \MSMsptruemuE        & & \MSMspThreemuEMean      & \MSMspThreemuEStD      & \MSMspThreemuEAStE      & & \MSMspFourmuEMean      & \MSMspFourmuEStD      & \MSMspFourmuEAStE      \\
			$\sigma^e$ & \MSMsptruesigmaE     & & \MSMspThreesigmaEMean   & \MSMspThreesigmaEStD   & \MSMspThreesigmaEAStE   & & \MSMspFoursigmaEMean   & \MSMspFoursigmaEStD   & \MSMspFoursigmaEAStE   \\
			$\rho$     & \MSMsptruerho        & & \MSMspThreerhoMean      & \MSMspThreerhoStD      & \MSMspThreerhoAStE      & & \MSMspFourrhoMean      & \MSMspFourrhoStD      & \MSMspFourrhoAStE      \\ 
			\cmidrule(lr){4-6} \cmidrule(l){8-10} 
			
			$\np$      & \multicolumn{1}{c}{} & \multicolumn{1}{c}{} & \multicolumn{3}{c}{\MSMspThreeI} & & \multicolumn{3}{c}{\MSMspFourI} \\
			$N$        & \multicolumn{1}{c}{} & \multicolumn{1}{c}{} & \multicolumn{3}{c}{\MSMspThreeN} & & \multicolumn{3}{c}{\MSMspFourN} \\
			$\nsims$   & \multicolumn{1}{c}{} & \multicolumn{1}{c}{} & \multicolumn{3}{c}{\MSMspThreeR} & & \multicolumn{3}{c}{\MSMspFourR} \\  
			time (sec.)& \multicolumn{1}{c}{} & \multicolumn{1}{c}{} & \multicolumn{3}{c}{\MSMspThreetime}& & \multicolumn{3}{c}{\MSMspFourtime} \\ 
			\bottomrule
		\end{tabular}
	\end{center}
	\vspace*{3mm}
	\scriptsize{	
		Note: This table reports the means, standard deviations (StD), and average asymptotic standard errors (AStE) of MSM estimates across 100 Monte Carlo samples for four simulation designs, \MCone--\MCfour. The true values of the estimated parameters are fixed across designs and listed in the second column. The covariates $\vX=(\X,\ldots,\X)$, with $\X$ scalar ($K=1$) and standard lognormal. The designs vary in the number of players $\np$, sample size $N$, and number of simulation draws $\nsims$ used by the estimator. \computeplatform 
	}
\end{table}

We explore the finite-sample performance of the MSM estimator in Monte Carlo experiments. We use two designs, \MCone\ and \MCtwo, that mirror those for MSL (but with lower $\nsims$) and two designs, \MCthree\ and \MCfour, for much larger problems than can be handled by MSL. To ensure the optimization procedure avoids local minima, we initialize it with a derivative-free global grid search for the best 5 of 500 random starting values, followed by a multi-start sequential quadratic programming procedure. 

Panel \MCone\ of Table \ref{table:MSM} reports the MSM estimates of a two-player game with $N=\MSMspOneN$ observations and $\nsims=\MSMspOneR$ simulation draws. Even in this small design, the point estimates show little bias and the asymptotic standard errors are close to the standard deviations across Monte Carlo samples. This suggests that the choice of moments and sandwich estimator work well in this baseline case.

In Panel \MCtwo, we increase the number of players to $\np=\MSMspTwoI$, with $N=\MSMspTwoN$ and $\nsims=\MSMspTwoR$. Relative to Design \MCone, the estimator becomes more precise for the parameters that are most directly tied to the interaction structure, while computation times increase only modestly. This is consistent with larger games generating more informative variation in the sequence and timing of exits.

In Panel \MCthree, we consider much larger samples ($N=\MSMspThreeN$), while reducing the number of simulation draws at $\nsims=\MSMspThreeR$. As expected, the standard deviations decrease relative to the smaller-sample designs, and the asymptotic standard errors continue to track them closely. At the same time, the estimates remain close to the true parameter values, indicating that the estimator performs well even when the number of simulation draws per observation is kept small.\footnote{In \citet[][Section S5.3]{ecsgsupp26:abbringyu}, we show that, even in small problems ($\np=2$ and $N=500$), a moderate number $\nsims$ of simulation draws suffices.} 

Finally, in Panel \MCfour, we report results for a larger game with $\np=\MSMspFourI$ players, $N=\MSMspFourN$ observations, and $\nsims=\MSMspFourR$ simulation draws. The estimator remains stable in this higher-dimensional setting, with small biases, standard deviations across Monte Carlo samples that are well aligned with the asymptotic standard errors, and manageable computation times. Compared with the MSL results, this suggests that MSM gives up some efficiency in small problems, but scales much more easily as the numbers of players and games grow.

Overall, Table \ref{table:MSM} indicates that the MSM estimator can recover the structural parameters of the model with small bias and reasonable precision across a range of designs, while keeping computation feasible in settings that are difficult for MSL. 

\section{Conclusion}
\label{s:conclusion}

We have put the econometrics of mixed hitting-time models to work on optimal stopping games with strategic complements, synchronization games. In particular, we have shown  that the latent L\'{e}vy process that drives the payoffs from stopping, the covariates' effect on the thresholds that characterize the game's equilibrium outcomes, and the joint distribution of the unobserved heterogeneity in these same thresholds can be identified and estimated from data on durations and covariates under fairly weak conditions. This provides a new approach to disentangling strategic interactions and other sources of synchronization, common shocks and unobserved heterogeneity, in optimal stopping.

This empirical approach can be adapted to stopping games with strategic substitutes, such as pre-emption games and wars of attrition, and cooperative stopping games. These will bring their own challenges, such as different types of equilibrium multiplicity and a need for mixed strategies, but will still allow equilibrium duration outcomes to be characterized in terms of easy-to-analyze thresholds. 

\appendix
\section*{Appendix}

\section{Characterization of the Data}\label{app:characterization}

\begin{proof}[Proof of Lemma \ref{lemma:characterization}]
Recall that we can represent $\mathbf{T}$ as  $\left(W^*; T^1,\ldots, T^{W^*};{\cal E}^1,\ldots,{\cal E}^{W^*-1}\right)$; where either $T^{W^*}<\infty$, $W^*=W$, and ${\cal E}^{W^*}={\cal I}\setminus\bigcup_{w=1}^{W^*-1}{\cal E}^w$; or $T^{W^*-1}<\infty$ but $T^{W^*}=\infty$. Because $T^w=\sum_{j=1}^w\Delta T^j$, with $T^0\equiv 0$, we can also represent $\mathbf{T}$ as  $\left(W^*;\Delta T^1,\ldots,\Delta T^{W^*};{\cal E}^1,\ldots,{\cal E}^{W^*-1}\right)$. Define an ordered partial partition $\mathbf{e}^w$ of ${\cal I}$ to be {\em complete} if $\bigcup_{j=1}^w e^j={\cal I}$ and {\em strictly partial} otherwise. Then, given $\vx$, the distribution of $\mathbf{T}$ can equivalently be characterized by  
\begin{enumerate}
\item the distribution of $\left(\Delta T^1,\ldots,\Delta T^{W^*}\right)$ on $\{W^*=w\}\bigcap\{{\cal E}^1=e^1,\ldots,{\cal E}^{w-1}=e^{w-1}\}$;
\item the distributions of $\left(\Delta T^1,\ldots,\Delta T^{w}\right)$ on
\begin{enumerate}
\item $\left\{T^{w}<\infty\right\}\bigcap\left\{{\cal E}^1=e^1,\ldots,{\cal E}^{w-1}=e^{w-1},{\cal E}^{w}={\cal I}\setminus\bigcup_{j=1}^{w-1}e^j\right\}$ and 
\item $\left\{T^{w}=\infty\right\}\bigcap\left\{{\cal E}^1=e^1,\ldots,{\cal E}^{w-1}=e^{w-1}\right\}$; and
\end{enumerate}
\item the Laplace transforms of these distributions,
\begin{equation}
\label{eq:lapTxcomplete}
\hspace*{-3em}\mathbf{\s}^{w}\mapsto\Lap^{w}\left(\mathbf{\s}^{w}\; ;\left(\mathbf{e}^{w-1},{\cal I}\setminus\bigcup_{j=1}^{w-1}e^j\right) \mid \vx\right)
\end{equation}
and
\begin{align}
\nonumber &\hspace*{-3em}1-\sum_{e^{1}\subseteq{\cal I}}\Lap^{1}\left(0\; ;e^{1} \mid \vx\right)&\text{ if }w=1,\\
\label{eq:lapTxdefect}
&\hspace*{-3em}\mathbf{\s}^{w-1}\mapsto\Lap^{w-1}\left(\mathbf{\s}^{w-1};\mathbf{e}^{w-1} \mid \vx\right)-\sum_{e^{w}\subseteq{\cal I}\setminus\bigcup_{j=1}^{w-1}e^j}\Lap^{w}\left(\left(\mathbf{\s}^{w-1},0\right);\left(\mathbf{e}^{w-1},e^{w}\right) \mid \vx\right)&\text{ if }w\geq 2;
\end{align}
\end{enumerate}

\noindent for all ordered strictly partial partitions $\mathbf{e}^{w-1}$ of ${\cal I}$ and $w=1,\ldots,\np$.

Clearly, we can construct \eqref{eq:lapTxcomplete} and \eqref{eq:lapTxdefect}, and thus the distribution of $\mathbf{T}\mid\vx$, if we know $\Lap^w\left(\cdot\; ;\mathbf{e}^w \mid \vx\right)$ for all ordered partial partitions $\mathbf{e}^w$ of ${\cal I}$; $w=1,\ldots,\np$. 

Conversely, \eqref{eq:lapTxcomplete} immediately gives $\Lap^w(\cdot\; ;\mathbf{e}^w\mid\vx)$ for all ordered complete partitions $\mathbf{e}^w$ of ${\cal I}$, which necessarily have $e^w={\cal I}\setminus\bigcup_{j=1}^{w-1}e^j$. Because all ordered partial partitions $\mathbf{e}^\np$ are complete, this gives $\Lap^\np(\cdot\; ;\mathbf{e}^\np\mid\vx)$ for all ordered partial partitions $\mathbf{e}^\np$ of ${\cal I}$. Using \eqref{eq:lapTxdefect} with $w=\np$, we can then construct $\Lap^{\np-1}(\cdot\; ;\mathbf{e}^{\np-1}\mid\vx)$ for all ordered strictly partial partitions $\mathbf{e}^{\np-1}$ of ${\cal I}$. Iterating this procedure, we can construct $\Lap^{w-1}(\cdot\; ;\mathbf{e}^{w-1}\mid\vx)$ for all ordered strictly partial partitions $\mathbf{e}^{w-1}$ of ${\cal I}$; $w=\np,\ldots,2$.
\end{proof}


\section{Identification with More than Two Players}
	\label{app:identGen}

\begin{proof}[Proof of Theorem \ref{th:identGen}]
By \eqref{eq:LapOne}, $\Lap_{\Delta T^{1}}(\s \mid (\x,\ldots,\x)) = \exp\left[ -\phi^{\np}(\x) \Lambda(\s) \right] \Lap_{\hat\err^\np}\left( \Lambda(\s) \right)$, where $\hat\err^\np\equiv\min_{i\in{\cal I}}\err_i$. Using this, and Assumption \ref{ass:variationGen} instead of Assumption \ref{ass:variation}, the first part of Lemma \ref{lemma:identTwo}'s proof establishes identification of $c^{-1}\Lambda$ and $c\phi^{\np}$. The identification of
\begin{equation}
\label{eq:lapYgen}
	\left(\salt^1,\ldots,\salt^w\right)\in\Rp^w\mapsto
	\mathbb{E}\left[\exp\left(-\sum_{j=1}^w\salt^jc\Delta\barYw^j(\vx)\right)\mathds{1}_{E^w(\mathbf{e}^w;\vx)}(\bigerr)\right]
\end{equation}
\noindent follows along the lines of the first part of Theorem \ref{th:identTwo}'s proof. 

Next, we recursively identify $c\phi^{S}$ for $S=\np-1,\ldots,1$. In step $S$ of this recursion, suppose that $c\phi^\np,\ldots,c\phi^{S+1}$ have been identified (note that we have already established this for 
the initial step, $S=\np-1$). Consider the partial sequential exit pattern $(\mathbf{e}^{w-2},\{i\})$, where $\mathbf{e}^{w-2}\equiv\left(\{1\},\ldots,\{w-2\}\right)$, player $i\geq w-1$, and $w\equiv\np-S+1$. By \eqref{eq:lapYgen}, we can identify the distribution of $\left(c\Delta\barYw^1(\vx),\ldots,c\Delta\barYw^w(\vx)\right)$ on $\left\{\bigerr\in E^{w-1}\left((\mathbf{e}^{w-2},\{i\});\vx\right)\right\}$ and thus, for $\zeta\in\Rp$ and covariates $\vx_i$ such that $\x_i=\tilde\x^{S+1}$ and $\x_j=\tilde\x^S$ for all $j\in\{w+1,\ldots,\np\}\setminus\{i\}$,
\begin{align*}
P^S_i\left(\zeta\mid\vx_i\right)
	&\equiv \Pr\left(c\Delta\barYw^w(\vx_i)>\zeta,\bigerr\in E^{w-1}\left((\mathbf{e}^{w-2},\{i\});\vx_i\right)\right)\\
	&= \Pr\bigl(\phi^\np(\x_1)+\err_1<\cdots<\phi^{S+2}(\x_w)+\err_w<\\
	&\hspace*{11.5em}\phi^{S+1}(\tilde\x^{S+1})+\err_i<\phi^S(\tilde\x^S)+\hat\err^S_{-i}-\zeta/c\bigr),
\end{align*}
where $\hat\err^S_{-i}\equiv\min_{j\in\{w+1,\ldots,\np\}\setminus\{i\}}\err_j$.
By Assumption \ref{ass:overlapGen}, we can pick $\tilde\x^{S+1}$ (which may differ from the previous iteration's $\tilde\x^S$) and $\tilde\x^S$ in $\vx_i$ so that $\sum_{i=\np-S}^{\np} P^S_i\left(0\mid\vx_i\right)\geq \Pr\left(\bigerr\in E^{w-2}(\mathbf{e}^{w-2};\vx_i)\right)$, corresponding to $\phi^{S+1}(\tilde\x^{S+1})\leq\phi^S(\tilde\x^S)$. By Assumption \ref{ass:cont}, $\sum_{i=\np-S}^{\np} P^S_i\left(\zeta\mid\vx_i\right)$ strictly decreases in $\zeta$ and equals $\Pr\left(\bigerr\in E^{w-2}\left(\mathbf{e}^{w-2};\vx_i\right)\right)$ at $\zeta=c\phi^S(\tilde\x^S)-c\phi^{S+1}(\tilde\x^{S+1})\geq 0$. As we know $c\phi^{S+1}(\tilde\x^{S+1})$, this identifies $c\phi^S(\tilde\x^S)$. Using this, and $P_i^S$ instead of $P_A$, we can identify $c\phi^S$ as in Theorem \ref{th:identTwo}'s proof.
 
Finally, to identify $G$, consider the set ${\cal E}_{seq}$ of all complete sequential exit patterns (all permutations of $\left(\{1\},\ldots,\{\np\}\right)$). Pick  $x^1,\ldots,x^\np\in\sX$ such that $\phi^{\np}(x^\np) \leq\cdots\leq \phi^{1}(x^1)$ (which is possible by Assumption \ref{ass:overlapGen}). Given $c\phi^{\np}(\x^\np),\ldots,c\phi^1(\x^1)$; the distributions of $\left(c\Delta\barYw^1(\vx_\mathbf{e}),\ldots,\Delta\barYw^\np(\vx_\mathbf{e})\right)$ on $\left\{\bigerr\in E^\np(\mathbf{e};\vx_\mathbf{e})\right\}$ for all $\mathbf{e}\in{\cal E}_{seq}$ and corresponding $\vx_\mathbf{e}$ that endow the player exiting in wave $\np-S+1$ with covariates $\x^S$; $S=\np,\ldots,1$; identify the distribution of $c\bigerr$ on $\left\{\bigerr\in \bigcup_{\mathbf{e}\in{\cal E}_{seq}} E^\np(\mathbf{e};\vx_\mathbf{e})\right\}$. As Assumption \ref{ass:overlapGen} ensures that $\Pr\left(\bigerr\in \bigcup_{\mathbf{e}\in{\cal E}_{seq}} E^\np(\mathbf{e};\vx_\mathbf{e})\right)=1$, this identifies $G(\cdot/c)$.
\end{proof}

\bibliographystyle{chicago}
\bibliography{ecsg}

\end{document}